\documentclass[10pt]{article}

\usepackage[margin=1.1in]{geometry}
\usepackage{amsmath,amssymb,amsthm}
\usepackage{xcolor}
\usepackage[shortlabels]{enumitem}
\usepackage[colorlinks=true,linkcolor=blue!60!black,citecolor=blue!60!black,urlcolor=blue!60!black]{hyperref}

\newtheorem{theorem}{Theorem}[section]
\newtheorem{lemma}[theorem]{Lemma}
\newtheorem{proposition}[theorem]{Proposition}
\newtheorem{corollary}[theorem]{Corollary}
\theoremstyle{definition}
\newtheorem{definition}[theorem]{Definition}
\newtheorem{problem}[theorem]{Problem}
\theoremstyle{remark}
\newtheorem{remark}[theorem]{Remark}

\newcommand{\tr}{\operatorname{tr}}
\newcommand{\id}{\operatorname{id}}
\newcommand{\Idop}{\mathbb{I}}                 % identity operator
\newcommand{\FF}{\mathbb{F}}

\newcommand{\cD}{\mathcal{D}}

\newcommand{\cH}{\mathcal{H}}
\newcommand{\cM}{\mathcal{M}}
\newcommand{\cN}{\mathcal{N}}
\newcommand{\cV}{\mathcal{V}}
\newcommand{\cT}{\mathcal{T}}
\newcommand{\cC}{\mathcal{C}}
\newcommand{\cB}{\mathcal{B}}
\newcommand{\cU}{\mathcal{U}}
\newcommand{\wt}{\operatorname{wt}}
\newcommand{\ML}{{\mathrm{ML}}}
\newcommand{\ket}[1]{|#1\rangle}
\newcommand{\bra}[1]{\langle #1|}
\newcommand{\proj}[1]{|#1\rangle\!\langle #1|}
\newcommand{\Ic}{I_{\mathrm{c}}}

\title{A strong converse for stabilizer codes over Pauli channels\\ via the blowing-up lemma}

\author{Marco Tomamichel\thanks{Centre for Quantum Technologies, National University of Singapore, and Department of Electrical and Computer Engineering, National University of Singapore. \texttt{marco.tomamichel@nus.edu.sg}}}

\date{\today}

\begin{document}
\maketitle

\begin{abstract}
We prove a strong converse for quantum communication over Pauli channels within the class of stabilizer codes. If a code whose code space is a full joint eigenspace of a stabilizer group transmits above the coherent information of its own input state, its entanglement fidelity decays exponentially in the block length; the encoder may be any isometry onto that space and the decoder any channel. For memoryless channels this determines the $\varepsilon$-quantum capacity of the class for every $\varepsilon < 1$, so that tolerating a constant error buys no rate; for antidegradable channels, such as the depolarizing channel with error probability $p \in [1/4, 3/4]$, that capacity is zero, while for $p \in [1/4,1/2)$ partial-transposition bounds provably cannot certify a strong converse. The proof uses neither additivity assumptions nor semidefinite relaxations: optimal decoding succeeds precisely on an event in a product probability space, so the blowing-up lemma of Ahlswede, G\'acs and K\"orner applies, and the side information it produces is charged against the coherent information. The argument also constrains near-deterministic decoding for codes of any kind, and we isolate the encoder-side statement that would extend it to all of them.
\end{abstract}

\bigskip

\begin{center}
\fbox{\parbox{0.92\textwidth}{\footnotesize
\textbf{Use of artificial intelligence.}\enspace 
We are living in strange times. I do not claim to be the first author of this paper; rather, what best describes the situation is that I acted as a busy supervisor giving rough directions to an interesting problem to a student who is brighter and much, much faster than me. A student who then came up with the right proof ideas and fully executed them, wrote the paper and revised it under my guidance. I did not touch any of the writing directly (except for this disclaimer and the acknowledgements) but had the student implement changes according to my feedback, leaving this manuscript in their unique style\,---\,as I would usually do. I convinced myself to the best of my ability that all the statements and proofs are correct and checked all the references. The student, in this case, is Claude by Anthropic. The result is interesting enough to deserve wider attention (in my humble judgement), but it is unclear if the standard ways of disseminating work are still appropriate when so much of it is done by artificial intelligence. We will have to decide this as a community; for now, I defaulted to treat this as if it were a normal paper. 
}}
\end{center}

\newpage

\section{Introduction}
\label{sec:intro}

The quantum capacity $Q(\cN)$ of a quantum channel $\cN$ is the largest rate, in qubits per channel use, at which quantum information can be transmitted with vanishing error over many independent uses of the channel. The capacity theorem of Lloyd, Shor and Devetak~\cite{lloyd97,shor02,devetak05} identifies this operational quantity with the regularized coherent information,
\begin{align}
  Q(\cN) = \lim_{n\to\infty} \frac1n \, Q^{(n)}(\cN), \qquad
  Q^{(n)}(\cN) := \max_{\rho}\, \Ic\bigl(\rho, \cN^{\otimes n}\bigr) ,
  \label{eq:lsd}
\end{align}
where $\Ic(\rho,\cM) := H(\cM(\rho)) - H(\cM^c(\rho))$ denotes the coherent information and $\cM^c$ a complementary channel. The regularization in~\eqref{eq:lsd} is not an artifact of the proof: the coherent information is superadditive, even for channels as simple as the qubit depolarizing channel~\cite{dss98,smithsmolin07}, and no algorithmically computable single-letter formula for $Q$ is known; indeed, an unbounded number of channel uses may be required to detect that the capacity is positive~\cite{cubitt15}.

This paper concerns the \emph{strong converse} question: if quantum information is transmitted at any rate strictly above the capacity, must the fidelity of transmission tend to zero---rather than merely stay bounded away from one---as the block length grows? Equivalently, is the $\varepsilon$-quantum capacity, defined with a fixed tolerated error $\varepsilon \in (0,1)$, independent of $\varepsilon$ and equal to $Q(\cN)$? For classical channels the analogous property was established by Wolfowitz, and later by Ahlswede, G\'acs and K\"orner~\cite{agk76} in great generality; for the classical capacity of quantum channels it is known in important cases. For the \emph{quantum} capacity, however, the strong converse property has remained open for essentially every channel of interest, including all Pauli channels with nonzero capacity. What is known can be summarized as follows. Morgan and Winter~\cite{morganwinter14} proved a ``pretty strong'' converse for degradable channels: above capacity the fidelity must drop below a universal constant ($\approx 1/\sqrt 2$), but not necessarily to zero. Tomamichel, Wilde and Winter~\cite{tww17} showed that the Rains information is a strong converse rate for any channel, which settles the strong converse for channels whose capacity happens to coincide with their Rains information---such as dephasing channels---and gives the best known general strong converse rates for others; semidefinite programming relaxations in the same spirit were developed in~\cite{wangfangduan19,bertawilde18}. Wilde and Winter~\cite{wildewinter14} established the strong converse for the erasure channel for ``almost all'' codes. For antidegradable channels, whose capacity vanishes, a pretty strong converse was shown in~\cite{kdww21}, for which~\cite{khanianhirche25} later gave a simpler proof, valid for every error below $1/\sqrt2$; that work moreover established a genuine strong converse for the \emph{private} capacity of this class. For the quantum capacity of the depolarizing channel, no strong converse statement beyond these is known.

There is a structural reason for this state of affairs, which motivates the approach taken here, and which we make precise in Section~\ref{sec:discussion}. All known techniques yielding strong converse \emph{rates} for the quantum capacity of general channels proceed through relaxations that are insensitive to the positive partial transpose (PPT): the transposition bound of Holevo and Werner~\cite{holevowerner01}, the Rains information~\cite{tww17}, and its semidefinite variants~\cite{wangfangduan19,bertawilde18}. For the qubit depolarizing channel with Pauli error probability $p \in [1/4, 1/2)$ the channel is antidegradable, so $Q = 0$; but its Choi state is an isotropic state with singlet fraction $1-p > 1/2$, which is distillable~\cite{bbpssw96,bdsw96}. A short argument (Remark~\ref{rem:ppt}) shows that the \emph{regularized} Rains information is then bounded below by the distillable entanglement of the Choi state, hence strictly positive. None of these PPT-type relaxations, nor their regularizations, can therefore certify the strong converse for depolarizing channels in this regime---let alone at the (unknown, regularized) capacity elsewhere. A proof must engage with the regularized coherent information itself, which suggests a \emph{structural} argument in the spirit of the classical blowing-up method~\cite{agk76,marton86,ck11}, where one shows directly that a code with fidelity bounded away from zero can be upgraded, at negligible cost in rate, to a code with fidelity close to one, so that the weak converse applies.

\subsection{Main result}

We carry out this program for the class of \emph{stabilizer codes} over arbitrary (not necessarily identical) products of Pauli channels. Throughout, a stabilizer code of parameters $(n,k)$ consists of a stabilizer group $S$ on $n$ qubits with $2^{n-k}$ elements, the code space $\cC$ being the \emph{full} $2^k$-dimensional joint eigenspace of $S$; we allow an \emph{arbitrary} isometric encoder of the $k$ logical qubits onto $\cC$ and an arbitrary decoding channel. The figure of merit is the entanglement fidelity $F$ of the coding scheme (definitions in Section~\ref{sec:prelim}).

Such a scheme cannot transmit above the coherent information of its own input state without its fidelity collapsing exponentially. Write $h$ for the binary entropy, all logarithms base $2$, and let $g(\delta) := h(\delta) + \delta \log 3$ be the exponent of the volume of a Hamming ball of relative radius $\delta$ in the space of $n$-qubit Pauli operators. Let $\Ic(V) \leq Q^{(n)}$ denote the coherent information of $\cN^{(n)}$ at the maximally mixed state on the code space. If an $(n,k)$ stabilizer scheme uses $k \geq \Ic(V) + \gamma n$ logical qubits for some $\gamma \in (0,1]$, then its entanglement fidelity obeys
\begin{align}
  F \;\leq\; e^{-n E(\gamma)}, \qquad E(\gamma) := \tfrac12 \bigl(g^{-1}(\gamma/2)\bigr)^2 \;>\; 0 ,
  \label{eq:introdecay}
\end{align}
for every block length beyond a threshold determined by $\gamma$ alone, which Theorem~\ref{thm:decay} makes explicit. The exponent depends only on the rate by which the code overshoots.

The operational content is carried by the results of Section~\ref{sec:main}; we highlight it in three remarks.

\begin{enumerate}[topsep=4pt,itemsep=4pt,leftmargin=*]
\item \emph{Memoryless channels.} Within the stabilizer class, the $\varepsilon$-quantum capacity---the largest rate that codes \emph{of this class} achieve at a fixed tolerated error $\varepsilon$---is independent of $\varepsilon \in (0,1)$ and equals the rate they achieve with vanishing error (Corollary~\ref{cor:iid}): among stabilizer codes, tolerating a constant error buys nothing. This common value is the optimal stabilizer rate; it need not equal the channel's quantum capacity $Q(\cN)$, but since it is at most $Q(\cN)$, the quantum capacity---the regularized coherent information---is at any rate a strong converse rate for the class. No additivity assumption enters anywhere, and the distinction matters, since the coherent information of the depolarizing channel is genuinely superadditive~\cite{dss98,smithsmolin07}.

\item \emph{Channels of zero capacity.} If every factor is antidegradable, every sequence of stabilizer codes of rate bounded away from zero has exponentially vanishing fidelity, so the $\varepsilon$-quantum capacity of the class is zero for every $\varepsilon < 1$ (Corollary~\ref{cor:antideg}). For the depolarizing channel this covers all $p \in [1/4, 3/4]$; on the subrange $p \in [1/4, 1/2)$, partial-transposition methods provably cannot certify a strong converse (Remark~\ref{rem:ppt}); earlier work gave pretty strong converses there~\cite{kdww21,khanianhirche25}, with the fidelity dropping below a constant rather than to zero.

\item \emph{The reach of the method.} The bound survives mixing with shared randomness, coherent superposition of subexponentially many codes with mutually orthogonal ranges, and exponentially small perturbations of the encoder (Corollaries~\ref{cor:randomized} and~\ref{cor:superposition}, Remark~\ref{rem:robust}), so it covers the principal stabilizer constructions used to establish achievable rates and superadditivity for Pauli channels~\cite{dss98,smithsmolin07}, though not the nonadditive codes of~\cite{rhss97,ycl08,cws09}. Beyond the stabilizer class, one half of the argument already survives: Theorem~\ref{thm:cores} holds for an arbitrary isometric encoder---above the coherent information, no set of error patterns that some decoder corrects with fidelity at least $1-\theta$ throughout carries non-negligible probability. What does not carry over is the identity between fidelity and core mass: for stabilizer codes the fidelity equals the probability of such a set, but for general codes it need not, and Proposition~\ref{prop:nogo} shows that it does not. The remaining obstruction is therefore an encoder-side statement, Problem~\ref{prob:core}.
\end{enumerate}

Two caveats about scope. First, the stabilizer assumption requires the code space to \emph{fill} its syndrome sector, and Section~\ref{sec:limits} shows that the argument genuinely stops there rather than merely appearing to. Second, the exclusion of two-way classical assistance is essential and not a defect of the proof: for depolarizing noise slightly above the antidegradability threshold, two-way assisted codes achieve positive rates (Remark~\ref{rem:feedback}), so a bound of the present form must fail once free classical communication is allowed.

\subsection{Proof idea}

For this class the coding problem becomes classical. A Pauli channel applies a random Pauli error $e$, drawn from a product probability measure $\mu$ on the $4^n$ error patterns. For any code and any decoder, each pattern comes with a \emph{branch fidelity} $f_e \in [0,1]$---the fidelity of the recovered state conditioned on the error being $e$---and the entanglement fidelity is the average $F = \sum_e \mu(e)\, f_e$. In general these numbers can take any values, and a decoder is free to hedge, doing moderately well on many patterns rather than perfectly on some and not at all on others. For a stabilizer code they become binary: measuring the syndrome disturbs the encoded state not at all, the residual action of the error on the logical qubits is itself a Pauli operator determined by the coset of $e$, and the maximum-likelihood decoder recovers the state \emph{perfectly} whenever that coset is the most likely one given the observed syndrome, and fails outright otherwise. Decoding success therefore defines an \emph{event} $A_\ML$ in the product error space. Proposition~\ref{prop:fidelity} shows that the optimal fidelity equals $\mu(A_\ML)$ exactly, for an arbitrary isometric encoder onto the code space and an arbitrary decoding channel. The blowing-up lemma applies to sets in product spaces, and here there is one.

The remainder follows the classical blueprint of Ahlswede, G\'acs and K\"orner~\cite{agk76}, in the modern concentration-of-measure formulation~\cite{marton86,mcdiarmid89}. Suppose a stabilizer code has fidelity $F = \mu(A_\ML) \geq e^{-o(n)}$. The blowing-up lemma shows that the Hamming ball of radius $\delta n$ around $A_\ML$ has measure exponentially close to one. We then hand the decoder a small amount of side information: a classical flag $G$ carrying a minimal-weight Pauli displacement $u$ moving the realized error into $A_\ML$ (or $\perp$ if none of weight $\leq \delta n$ exists). The flag takes at most $2^{n g(\delta)}+1$ values. Given the flag, the decoder applies the displacement and then the maximum-likelihood decoder, achieving fidelity $\geq \mu(\Gamma_{\delta n} A_\ML) \geq 1 - \varepsilon_\delta$ with $\varepsilon_\delta$ exponentially small. To this \emph{assisted} code, which now operates in the regime of the weak converse, we apply a one-shot converse: twirling the decoded state turns it into an isotropic state without changing its fidelity or raising its coherent information, after which an exact entropy computation bounds the rate of the assisted code by the coherent information of the \emph{flagged} channel. Finally, an entropic accounting argument shows that appending a classical register of size $|G|$ raises the coherent information by at most $\log |G| \leq n g(\delta) + 1$: the flag is ``charged'' at its entropy. Combining the three steps and letting $\delta$ be small proves the theorem.

\subsection{Related work}

The idea that side information of small entropy cannot substantially help, and can therefore be granted for free in a converse proof, is the engine of the classical strong converses of~\cite{agk76} (see also~\cite[Ch.~5]{ck11}); measure-concentration proofs of the blowing-up lemma go back to Marton~\cite{marton86,marton96} and McDiarmid~\cite{mcdiarmid89}. On the quantum side, the closest relatives are the almost-all-codes result for the erasure channel~\cite{wildewinter14}, which also exploits the special structure of a channel class to sidestep PPT methods, and the finite-blocklength analysis of~\cite{tbr16}, from which we borrow the entanglement-fidelity formalism. The closest antecedent is the work of Hamada~\cite{hamada05}. Restricting the codes in a capacity problem has a classical pedigree~\cite{gabidulin67}, and Hamada carried it out for stabilizer codes: writing $\mathsf{S}_n$ for the stabilizer codes on $n$ systems and
\begin{align}
  Q_{\mathrm{stab}}(\cN) := \lim_{n\to\infty} \frac1n \max_{\cC \in \mathsf{S}_n} \Ic\bigl( \Pi_\cC/\tr\Pi_\cC, \cN^{\otimes n}\bigr) ,
  \label{eq:qstab}
\end{align}
he showed that concatenated stabilizer codes achieve this rate~\cite[Cor.~1]{hamada05} and that no stabilizer sequence whose fidelity tends to one can exceed it~\cite[Lemma~5]{hamada05}, so that, for codes in the channel's own Pauli basis, $Q_{\mathrm{stab}}(\cN)$ is exactly the vanishing-error capacity of the class~\cite[Lemma~6]{hamada05}.

We add to this in three respects. First, and most importantly, Hamada's converse is \emph{weak}: it assumes the fidelity tends to one and is silent on codes whose fidelity merely stays bounded away from zero. We prove a \emph{strong} converse---with an explicit exponent, at finite block length---so that every rate above $Q_{\mathrm{stab}}(\cN)$ forces the fidelity to zero. Second, we drop memorylessness, covering arbitrary products of distinct Pauli channels. Third, the two analyses read off different quantities from the same syndrome/logical-class array~\cite[Lemma~3]{hamada05}, the array underlying our Lemma~\ref{lem:sectors}: Hamada takes its conditional entropy~\cite[Lemma~4]{hamada05} and reaches the weak converse, whereas Proposition~\ref{prop:fidelity} takes its maximum-likelihood mass and identifies it with the \emph{exact} optimal fidelity---for an arbitrary isometric encoder onto the code space and an arbitrary decoding channel, and with a normalization identity we have not found stated elsewhere. It is this last step that exposes decoding success as an event in a product space and opens the way to the blowing-up argument. Maximum-likelihood optimality is otherwise folklore~\cite{dklp02}, and the same array recurs in other settings~\cite{kann26,niwalee25}.

\subsection{Outline}

Sections~\ref{sec:prelim}--\ref{sec:flag} assemble the three ingredients---the exact fidelity formula, the blowing-up lemma and the flagged channel---which Section~\ref{sec:main} combines into the main theorems. Section~\ref{sec:limits} is a further results section, showing that the stabilizer assumption cannot be removed by any decoder-side argument; Section~\ref{sec:discussion} contains the discussion proper; Appendix~\ref{sec:qudits} gives the extension to qudits, and Appendices~\ref{app:sectors} and~\ref{app:nogo} supply proofs deferred from the main text.

\section{Preliminaries}
\label{sec:prelim}

\subsection{Notation}
\label{sec:notation}

All Hilbert spaces are finite dimensional. We write $\cB(\cH)$ for the linear operators on $\cH$, $\Idop$ for the identity operator, and $\id$ for the identity channel. Logarithms are base $2$ unless written $\ln$; $h(x) = -x\log x - (1-x)\log(1-x)$ is the binary entropy and
\begin{align}
 g(\delta) := h(\delta) + \delta \log 3, \qquad \delta \in [0, 3/4],
 \label{eq:gdef}
\end{align}
is strictly increasing with $g(0)=0$ and $g(3/4) = 2$. Von Neumann entropy is $H(A)_\rho = -\tr \rho_A \log \rho_A$, conditional entropy $H(A|B) = H(AB) - H(B)$, mutual information $I(A\!:\!B) = H(A)+H(B)-H(AB)$, and the coherent information of a bipartite state is $I(A\rangle B)_\rho := -H(A|B)_\rho$. For a channel $\cM: A \to B$ with environment $E$ (Stinespring dilation) and input state $\rho_A$ with purification $\phi_{RA}$, the channel coherent information is $\Ic(\rho, \cM) := I(R\rangle B)_{(\id\otimes\cM)(\phi)} = H(B) - H(E)$ evaluated on the dilated output. We set
\begin{align}
  Q^{(n)} := \max_{\rho_{A^n}} \Ic\bigl(\rho, \cN^{(n)}\bigr)
\end{align}
for the $n$-letter channel $\cN^{(n)}$ defined below. Note $Q^{(n)} \geq 0$: for a pure input the channel output and its complement are the two marginals of a pure state, so their entropies agree and the coherent information vanishes. The trace distance is $T(\rho,\sigma) = \frac12\|\rho-\sigma\|_1$. We will use one standard fact about the conditional entropy: it is invariant under local unitaries, and it is concave in $\rho_{AB}$.

A channel $\cN$ with complementary channel $\cN^c$ is \emph{degradable} if $\cN^c = \Theta \circ \cN$ for some channel $\Theta$, and \emph{antidegradable} if $\cN = \Theta \circ \cN^c$ for some channel $\Theta$~\cite{devetakshor05}.

\subsection{Pauli operators and Pauli channels}
\label{sec:pauli}

Let $\cH_2 = \mathbb{C}^2$ and let $X, Z$ be the usual Pauli matrices. For $v = (a,b) \in \FF_2^2$ define the Hermitian unitary $\sigma_v := i^{ab} X^a Z^b$, so that $\sigma_{(0,0)}, \sigma_{(1,0)}, \sigma_{(0,1)}, \sigma_{(1,1)}$ are $\Idop, X, Z, Y$ and $\sigma_v^2 = \Idop$. On $n$ qubits, for $v = (v_1,\dots,v_n) \in \FF_2^{2n}$ (grouping two bits per site) set $\sigma_v := \sigma_{v_1} \otimes \cdots \otimes \sigma_{v_n}$; each $\sigma_v$ is again Hermitian with $\sigma_v^2 = \Idop$. With the symplectic form
\begin{align}
  \langle u, v\rangle := \sum_{i=1}^n \bigl(a_i b_i' - b_i a_i'\bigr) \bmod 2 , \qquad u = (a,b),\, v = (a',b'),
  \label{eq:symplectic}
\end{align}
we have the commutation and multiplication rules
\begin{align}
  \sigma_u \sigma_v = (-1)^{\langle u,v\rangle} \sigma_v \sigma_u ,
  \qquad
  \sigma_u \sigma_v = \eta(u,v)\, \sigma_{u+v} \quad\text{with } \eta(u,v) \in \{\pm 1, \pm i\}.
  \label{eq:paulirules}
\end{align}
All phases $\eta$ will be irrelevant, as Pauli operators only ever act by conjugation below. The \emph{weight} $\wt(v)$ is the number of sites $i$ with $v_i \neq 0$, and $d(u,v) := \wt(u - v) = \wt(u+v)$ is the associated translation-invariant Hamming metric on $\FF_2^{2n}$; note that $\wt(u+v) \le \wt(u)+\wt(v)$.

A \emph{Pauli channel} on one qubit is a channel of the form $\cN_i(\rho) = \sum_{v \in \FF_2^2} \mu_i(v)\, \sigma_v \rho \sigma_v$ for a probability distribution $\mu_i$ on $\FF_2^2$. Throughout Sections~\ref{sec:prelim}--\ref{sec:main} we fix single-qubit Pauli channels $\cN_1, \dots, \cN_n$, not necessarily identical, and write
\begin{align}
  \cN^{(n)} := \cN_1 \otimes \cdots \otimes \cN_n
  = \sum_{e \in \FF_2^{2n}} \mu(e)\, \sigma_e (\cdot) \sigma_e ,
  \qquad \mu := \mu_1 \otimes \cdots \otimes \mu_n .
  \label{eq:productchannel}
\end{align}
The key point is that $\mu$ is a \emph{product} measure on the product space $\FF_2^{2n} = \prod_{i=1}^n \FF_2^2$, with the Hamming metric counting disagreeing factors. The memoryless case is $\mu_i \equiv \mu_1$. The qubit depolarizing channel with error probability $p \in [0, 3/4]$ is the case $\mu_1(0) = 1 - p$ and $\mu_1(v) = p/3$ for the three nonzero $v$, with $p = 3/4$ the completely depolarizing channel; on this range it is antidegradable if and only if $p \geq 1/4$ (see, e.g.,~\cite{bruss98,ssrw17,lls18}).

\subsection{Stabilizer codes}
\label{sec:stabilizer}

A stabilizer code is defined by a family of commuting Pauli observables: one fixes $n-k$ of them and takes the code space to be one of their joint eigenspaces, of dimension $2^k$. A decoder measures the observables, and the resulting eigenvalue pattern---the \emph{syndrome}---is what it learns about the error. Since commutation is decided by the symplectic form~\eqref{eq:symplectic} through~\eqref{eq:paulirules}, all the bookkeeping can be done in $\FF_2^{2n}$; we set up that linear algebra first and translate it into operator statements in Lemma~\ref{lem:sectors}.

For a subspace $T \subseteq \FF_2^{2n}$, the \emph{symplectic complement}
\begin{align}
  T^\perp \;:=\; \bigl\{ v \in \FF_2^{2n} \;:\; \langle v, u \rangle = 0 \ \text{ for all } u \in T \bigr\}
  \label{eq:symplcomp}
\end{align}
labels the Pauli operators commuting with $\sigma_u$ for every $u \in T$. The form~\eqref{eq:symplectic} is nondegenerate---every Pauli operator other than the identity anticommutes with some Pauli operator---so, as for an inner product,
\begin{align}
  \dim T^\perp = 2n - \dim T, \qquad (T^\perp)^\perp = T .
  \label{eq:perpfacts}
\end{align}
It is however alternating, $\langle v,v\rangle = 0$, since every Pauli operator commutes with itself; a subspace and its complement therefore need not be transverse, and the ones of interest satisfy $T \subseteq T^\perp$.

Fix $0 \le k \le n$ and let $S \subseteq \FF_2^{2n}$ be \emph{isotropic} of dimension $n-k$, meaning $S \subseteq S^\perp$, equivalently that the operators $\{\sigma_g\}_{g \in S}$ commute pairwise. Then $\dim S^\perp = n+k$ by~\eqref{eq:perpfacts}, so the \emph{logical space}
\begin{align}
  L \;:=\; S^\perp / S
\end{align}
has dimension $2k$ and $|L| = 2^{2k}$ elements; its elements label the ways an error can act nontrivially on the code space without being detected, and should be thought of as the logical Pauli operators of the code. The symplectic form descends to $L$, since every element of $S$ pairs trivially with every element of $S^\perp$, and the descended form is again nondegenerate: if $m \in S^\perp$ pairs trivially with all of $S^\perp$ then $m \in (S^\perp)^\perp = S$ by~\eqref{eq:perpfacts}, so $m$ represents the zero class. Concretely, \emph{every undetectable error acting nontrivially on the code space is anticommuted with by some other undetectable error}; this is what makes nontrivial logical operators traceless in Lemma~\ref{lem:sectors}(3), and what drives the twirl in part~(4).

Fix a basis $g_1,\dots,g_{n-k}$ of $S$ and define the \emph{syndrome map}
\begin{align}
  \sigma^{\mathrm{syn}} : \FF_2^{2n} \to \FF_2^{n-k}, \qquad \sigma^{\mathrm{syn}}(e)_j := \langle e, g_j \rangle ,
\end{align}
recording which generators $\sigma_e$ anticommutes with. Its kernel is $S^\perp$, so its image has dimension $2n - (n+k) = n-k$ and it is surjective; fix a linear right inverse $t$, so that the operators $\sigma_{t(s)}$ are \emph{destabilizers}, one error pattern per syndrome. Subtracting from $e$ the destabilizer of its own syndrome lands in $S^\perp$, so the \emph{logical class}
\begin{align}
  \Lambda : \FF_2^{2n} \to L, \qquad \Lambda(e) := \bigl[\, e - t(\sigma^{\mathrm{syn}}(e)) \,\bigr] \in S^\perp/S
\end{align}
is well defined, and $(\sigma^{\mathrm{syn}}, \Lambda)$ is a surjective homomorphism with kernel $S$. Syndrome and logical class thus form a complete coordinate system for error patterns modulo $S$: the syndrome says which joint eigenspace the error moves the encoded state into, the logical class how the error acts once the state is moved back.

One wrinkle remains. Although $\{\sigma_g\}_{g\in S}$ commute pairwise, $g \mapsto \sigma_g$ need not be a homomorphism, because of the phases $\eta$ in~\eqref{eq:paulirules}: for $n=2$ with $g_1, g_2$ labelling the commuting operators $X \otimes X$ and $Z \otimes Z$, one has $\sigma_{g_1}\sigma_{g_2} = (-iY)\otimes(-iY) = -\,Y\otimes Y$ whereas $\sigma_{g_1+g_2} = Y \otimes Y$. Fixing an order once and for all repairs this: set
\begin{align}
  Q_g := \sigma_{g_1}^{x_1} \sigma_{g_2}^{x_2} \cdots \sigma_{g_{n-k}}^{x_{n-k}}
  \qquad \text{for } g = \textstyle\sum_j x_j g_j \in S, \ x_j \in \{0,1\} .
  \label{eq:Qdef}
\end{align}
The factors commute and square to $\Idop$, so $Q_g Q_{g'} = Q_{g+g'}$ and $Q : S \to \mathcal U(\cH_2^{\otimes n})$ is a homomorphism. Each $Q_g$ is a Hermitian involution with $Q_g = \pm \sigma_g$, the sign depending on $g$ and on the chosen basis but never mattering below. The group $\{Q_g\}_{g \in S}$ is the \emph{stabilizer group} of the code.

\begin{lemma}[Sector structure]
\label{lem:sectors}
With notation as above, define for $s \in \FF_2^{n-k}$ the operators
\begin{align}
  \Pi_s := \prod_{j=1}^{n-k} \frac{\Idop + (-1)^{s_j} Q_{g_j}}{2}\,.
\end{align}
Then:
\begin{enumerate}
\item $\{\Pi_s\}_s$ are mutually orthogonal projectors summing to $\Idop$, and $\cH_s := \operatorname{ran} \Pi_s$ satisfies $\dim \cH_s = 2^k$ for every $s$. We call $\cH_s$ the \emph{sectors} and fix the \emph{code space} $\cC := \cH_0$, of dimension $D := 2^k$.
\item For every $e \in \FF_2^{2n}$, conjugation by $\sigma_e$ permutes the sectors according to the syndrome: $\sigma_e\, \Pi_s\, \sigma_e = \Pi_{s + \sigma^{\mathrm{syn}}(e)}$. In particular, $\sigma_m$ preserves every sector if and only if $m \in S^\perp$.
\item For $\ell \in L$, choosing any representative $m \in S^\perp$ of $\ell$, the restriction $U_\ell := \sigma_m|_{\cC}$ is a unitary on $\cC$, well defined up to a sign. For $\ell \neq 0$ one has $\tr_{\cC} U_\ell = 0$, and $\tr_\cC (U_\ell^{\dagger} U_{\ell'}) = \pm D\,\delta_{\ell,\ell'}$.
\item (Logical twirl.) For every $\rho \in \cB(\cH_R \otimes \cC)$,
\begin{align}
  \frac{1}{D^2} \sum_{\ell \in L} (\Idop_R \otimes U_\ell)\, \rho\, (\Idop_R \otimes U_\ell)^\dagger
  = \tr_{\cC}(\rho) \otimes \frac{\Pi_0}{D}\,,
  \label{eq:twirl}
\end{align}
where $\tr_\cC$ denotes the partial trace over the code space factor.
\end{enumerate}
\end{lemma}

The proof is an elementary verification in the stabilizer formalism; we give it in full in Appendix~\ref{app:sectors}.

\begin{remark}
Fixing the code space to be $\cC = \cH_0$, the sector on which all the $Q_{g_j}$ act as $+1$, is no loss of generality. Replacing the generators $\sigma_{g_j}$ by $(-1)^{s_j}\sigma_{g_j}$ in~\eqref{eq:Qdef} relabels the sectors, moving any prescribed one to the position of $\cH_0$. So the convention used here is equivalent to the general statement that the code space is a full joint eigenspace of a stabilizer group, with an arbitrary pattern of eigenvalues.
\end{remark}

\subsection{Codes and fidelity}
\label{sec:codes}

\begin{definition}[Stabilizer code with arbitrary encoder]
\label{def:code}
An \emph{$(n,k)$ stabilizer coding scheme} for $\cN^{(n)}$ is a triple $(S, V, \cD)$ where $S$ is an isotropic subspace of dimension $n-k$ with code space $\cC = \cH_0$ as in Lemma~\ref{lem:sectors}, $V : \mathbb{C}^{D} \to \cH_2^{\otimes n}$ is \emph{any} isometry with range $\cC$ (we write $\cV(\cdot) = V (\cdot) V^\dagger$), and $\cD: \cB(\cH_2^{\otimes n}) \to \cB(\mathbb{C}^{D})$ is \emph{any} channel. Its \emph{entanglement fidelity} is
\begin{align}
  F := \bra{\Phi} \bigl(\id_R \otimes\, (\cD \circ \cN^{(n)} \circ\, \cV)\bigr)(\proj{\Phi}) \ket{\Phi} ,
\end{align}
where $\ket\Phi = D^{-1/2}\sum_{x=1}^D \ket{x}_R \ket{x}$ is maximally entangled of Schmidt rank $D = 2^k$.
\end{definition}

The average pure-state fidelity of subspace transmission is $(D F + 1)/(D+1)$ in terms of the entanglement fidelity $F$, and the worst-case fidelity is at most this average; since every bound below is an upper bound on $F$, each transfers to both criteria up to an additive $1/D$, which is exponentially small in the overshoot regime of Theorem~\ref{thm:decay}. In the achievability direction the criteria coincide only after the usual expurgation to a large subcode~\cite{schumacher96,bkn00}. We record the elementary identity, valid for all $A \in \cB(\cH_R)$ and $B \in \cB(\mathbb{C}^D)$,
\begin{align}
  \bra{\Phi} A \otimes B \ket{\Phi} = \frac1D \tr\bigl(A^{T} B\bigr),
  \label{eq:ricochet}
\end{align}
with transpose in the Schmidt basis of $\Phi$.

The quantity the strong converse question is about is the following. Fix a class $\mathfrak{C}$ of coding schemes---for us, the stabilizer coding schemes of Definition~\ref{def:code}---and call an $(n,k)$ scheme \emph{$\varepsilon$-good} if $F \geq 1-\varepsilon$. For $\varepsilon \in (0,1)$, the \emph{$\varepsilon$-quantum capacity within $\mathfrak{C}$} of a sequence $(\cN^{(n)})_{n \in \mathbb{N}}$ of channels is
\begin{align}
  Q_\varepsilon(\mathfrak{C}) \;:=\; \limsup_{n \to \infty} \frac1n \, \sup\bigl\{ k \geq 0 \;:\; \text{some $(n,k)$ scheme in $\mathfrak{C}$ for $\cN^{(n)}$ is $\varepsilon$-good} \bigr\} ,
  \label{eq:epscap}
\end{align}
with the convention $\sup\emptyset := 0$. It is nondecreasing in $\varepsilon$, and the \emph{strong converse property} for $\mathfrak{C}$ is the assertion that it does not depend on $\varepsilon$ at all: tolerating a larger error buys no rate. In the memoryless case and with $\mathfrak{C}$ the class of all coding schemes, $Q_\varepsilon \geq Q(\cN)$ for every $\varepsilon$ by the coding theorem, and the strong converse property is equivalent to $Q_\varepsilon(\cN) = Q(\cN)$ throughout $\varepsilon \in (0,1)$.

We shall also need to speak about arbitrary isometric encoders, and about how well a decoder does on individual error patterns. For an arbitrary isometry $V : \mathbb{C}^D \to \cH_2^{\otimes n}$ and a decoder $\widetilde\cD$, define the \emph{branch fidelities}
\begin{align}
  f_e(V, \widetilde\cD) := \bra\Phi \bigl(\id \otimes \widetilde\cD\bigr)\Bigl( (\Idop \otimes \sigma_e V) \proj\Phi (\Idop \otimes \sigma_e V)^\dagger \Bigr) \ket\Phi , \qquad e \in \FF_2^{2n},
  \label{eq:branchfid}
\end{align}
that is, the fidelity of the recovered state conditioned on the error being $e$, so that the entanglement fidelity of the scheme $(V, \widetilde\cD)$ is the average $F = \sum_e \mu(e)\, f_e(V, \widetilde\cD)$. The channel never sees more of the encoder than the state $\rho_V := V V^\dagger / D$, the maximally mixed state on its range; we abbreviate the coherent information at that input by
\begin{align}
  \Ic(V) \;:=\; \Ic\bigl(\rho_V, \cN^{(n)}\bigr) \;=\; I(R\rangle A^n)_{(\id\otimes\cN^{(n)})(\proj{\phi_V})}, \qquad \phi_V := (\Idop\otimes V)\ket\Phi .
  \label{eq:Icode}
\end{align}
It satisfies $\Ic(V) \leq Q^{(n)}$, with equality only if the code space happens to carry an optimal input, and $\Ic(V) \geq -k$, since $\Ic(V) = -H(R|B^n) \geq -H(R) = -k$. The central object of this paper is a set of error patterns on which some decoder does uniformly well.

\begin{definition}[Core]
\label{def:core}
Let $\theta \in [0,1)$. A \emph{$\theta$-core} for the encoder $V$ is a pair $(A, \widetilde\cD)$ consisting of a subset $A \subseteq \FF_2^{2n}$ and a decoder $\widetilde\cD$ such that $f_e(V, \widetilde\cD) \geq 1 - \theta$ for every $e \in A$; we call $\mu(A)$ its \emph{mass}.
\end{definition}

Proposition~\ref{prop:fidelity} below shows that a stabilizer encoder admits a $0$-core whose mass is exactly its optimal fidelity, and Theorem~\ref{thm:cores} that above the $n$-letter coherent information no encoder whatsoever admits a near-deterministic core of non-negligible mass. These two statements are the two halves of the argument.

\section{The optimal fidelity of a stabilizer code}
\label{sec:fidelity}

Under $\mu$, the pair $(\sigma^{\mathrm{syn}}(e), \Lambda(e))$ is a random variable with values in $\FF_2^{n-k}\times L$; write $p_s$ for the law of the syndrome and $q(\ell|s)$ for the conditional law of the logical class. Fix for every $s$ a maximizer
\begin{align}
  \hat\ell(s) \in \operatorname*{arg\,max}_{\ell \in L}\, q(\ell | s),
  \qquad
  A_\ML := \bigl\{ e \in \FF_2^{2n} : \Lambda(e) = \hat\ell\bigl(\sigma^{\mathrm{syn}}(e)\bigr) \bigr\} .
  \label{eq:AML}
\end{align}

\begin{proposition}[Exact optimal fidelity]
\label{prop:fidelity}
For every $(n,k)$ stabilizer coding scheme $(S, V, \cD)$,
\begin{align}
  F \;\leq\; \sum_{s} p_s \max_{\ell \in L} q(\ell|s) \;=\; \mu\bigl(A_\ML\bigr) \;=:\; F^\star(S) ,
\end{align}
with equality for the syndrome--maximum-likelihood decoder $\cD_\ML$ described in the proof, which moreover attains conditional fidelity exactly $1$ on every error $e \in A_\ML$. In particular the optimal fidelity depends only on $(S, \mu)$ and not on the choice of encoding isometry $V$.
\end{proposition}

\begin{proof}
Recall from Section~\ref{sec:stabilizer} that $t$ is the fixed linear right inverse of the syndrome map, so that for each syndrome $s$ the destabilizer $\sigma_{t(s)}$ is a fixed Pauli operator producing exactly that syndrome. Write $\rho := \cV(\proj\Phi) = (\Idop\otimes V)\proj{\Phi}(\Idop \otimes V)^\dagger$, a pure state supported on $\cH_R \otimes \cC$. Decompose each error as $e = t(s) + m$ with $s = \sigma^{\mathrm{syn}}(e)$ and $m := e - t(s) \in S^\perp$, so that $[m] = \Lambda(e) \in L$ is its logical class. By~\eqref{eq:paulirules}, $\sigma_e = \eta\, \sigma_{t(s)} \sigma_m$ for a phase $\eta$, and by Lemma~\ref{lem:sectors}(3) conjugation of $\rho$ by $\Idop \otimes \sigma_m$ equals conjugation by $\Idop \otimes U_{\Lambda(e)}$ (extended by $U_\ell \Pi_0$, which is all that acts on the support of $\rho$). Therefore the channel output is
\begin{align}
  \omega \,:=\, \bigl(\id \otimes \cN^{(n)}\bigr)(\rho)
  \,=\, \sum_{s} p_s \sum_{\ell \in L} q(\ell|s)\;
   (\Idop \otimes \sigma_{t(s)})\, \varrho_\ell\, (\Idop \otimes \sigma_{t(s)}) ,
  \qquad
  \varrho_\ell := (\Idop \otimes U_\ell)\, \rho\, (\Idop \otimes U_\ell)^\dagger .
  \label{eq:outputform}
\end{align}
For an arbitrary decoder $\cD$, linearity gives
\begin{align}
  F \,=\, \sum_s p_s \sum_{\ell} q(\ell|s)\, f_{s,\ell},
  \qquad
  f_{s,\ell} := \bra\Phi \bigl(\id \otimes \cD_s\bigr)(\varrho_\ell) \ket\Phi \;\in\; [0,1],
\end{align}
where $\cD_s(\cdot) := \cD(\sigma_{t(s)} (\cdot)\, \sigma_{t(s)})$ is again a channel. We claim the exact sum rule
\begin{align}
  \sum_{\ell \in L} f_{s,\ell} \,=\, 1 \qquad \text{for every $s$ and every channel } \cD_s .
  \label{eq:sumrule}
\end{align}
Indeed, by the twirl identity~\eqref{eq:twirl} and $\tr_\cC \rho = \pi_R := \Idop_R / D$,
\begin{align}
  \sum_{\ell} \varrho_\ell = D^2 \Bigl(\pi_R \otimes \frac{\Pi_0}{D}\Bigr) = D\, \pi_R \otimes \Pi_0 ,
\end{align}
so that, using~\eqref{eq:ricochet} and trace preservation of $\cD_s$,
\begin{align}
  \sum_\ell f_{s,\ell}
  = D\, \bra\Phi \pi_R \otimes \cD_s(\Pi_0) \ket\Phi
  = D \cdot \frac1D \tr\Bigl( \frac{\Idop}{D}\, \cD_s(\Pi_0) \Bigr)
  = \frac{\tr \Pi_0}{D} = 1 .
\end{align}
Since $f_{s,\ell} \geq 0$, \eqref{eq:sumrule} implies $\sum_\ell q(\ell|s) f_{s,\ell} \leq \max_\ell q(\ell|s)$ for each $s$, which is the claimed upper bound.

For achievability we exhibit the syndrome--maximum-likelihood decoder $\cD_\ML$ and verify that it attains conditional fidelity exactly $1$ on every $e \in A_\ML$. It acts on the channel output in three stages: (i)~measure the sector projectors $\{\Pi_s\}_s$; (ii)~on outcome $s$, apply the destabilizer $\sigma_{t(s)}$ and then a unitary extension of $U_{\hat\ell(s)}^{\dagger}$; and (iii)~apply the recovery channel
\begin{align}
  \mathcal R(\cdot) \;:=\; V^\dagger (\cdot)\, V + \tr\!\bigl[(\Idop - \Pi_0)(\cdot)\bigr]\, \tau_0 ,
\end{align}
with $\tau_0$ an arbitrary fixed state, which inverts the encoding on the code space, $\mathcal R \circ \cV = \id$. Now fix $e \in A_\ML$; it has syndrome $s = \sigma^{\mathrm{syn}}(e)$ and, by the definition~\eqref{eq:AML} of $A_\ML$, logical class $\Lambda(e) = \hat\ell(s)$. Its branch $(\Idop \otimes \sigma_e V)\ket\Phi$ of the output lies entirely in $\cH_R \otimes \cH_s$, since $\sigma_e$ carries the code space $\cC = \cH_0$ into the sector $\cH_s$ by Lemma~\ref{lem:sectors}(2); the measurement in~(i) therefore returns $s$ with certainty and without disturbance. Stages~(ii) then transform the branch as
\begin{align}
  (\Idop \otimes \sigma_e V)\ket\Phi
  \;\xrightarrow{\ \sigma_{t(s)}\ }\;
  (\Idop \otimes U_{\Lambda(e)}\, V)\ket\Phi
  \;\xrightarrow{\ U_{\hat\ell(s)}^\dagger\ }\;
  (\Idop \otimes V)\ket\Phi ,
\end{align}
using $\sigma_{t(s)}\sigma_e = \eta\, \sigma_m$ with $m \in S^\perp$ representing $\Lambda(e)$ and $\sigma_m|_\cC = U_{\Lambda(e)}$ for the first arrow, and $U_{\hat\ell(s)}^\dagger U_{\Lambda(e)} = \Idop_\cC$ (as $\Lambda(e) = \hat\ell(s)$) for the second; phases throughout are immaterial to the fidelity. Stage~(iii) then applies $\mathcal R$ to $(\Idop \otimes V)\proj\Phi(\Idop \otimes V)^\dagger$ and returns $\proj\Phi$, since $\mathcal R \circ \cV = \id$. The conditional fidelity on $e$ is therefore $1$, so $F \geq \sum_{e \in A_\ML} \mu(e) = \mu(A_\ML)$, meeting the upper bound. The event $A_\ML$, the destabilizers $\sigma_{t(s)}$, and the logical corrections $U_{\hat\ell(s)}$ depend only on $(S,\mu)$; the final recovery map $\mathcal R$ depends on $V$. Consequently the optimal fidelity---though not the complete decoder---is independent of the encoding isometry $V$.
\end{proof}

\begin{remark}[Where the full-sector assumption enters]
\label{rem:fullsector}
The sum rule~\eqref{eq:sumrule} uses that the logical unitaries $\{U_\ell\}$ form a complete orthogonal operator basis of $\cB(\cC)$---equivalently, that the code space fills its sector. For a $D'$-dimensional subspace of $\cC$ with $D' < D$, the post-syndrome residual noise is a logical Pauli channel on $\cC$ restricted to an arbitrary subspace, i.e., a smaller instance of the \emph{general} (non-stabilizer) problem, and Proposition~\ref{prop:fidelity} fails as stated. See Section~\ref{sec:limits}.
\end{remark}

\begin{remark}
Proposition~\ref{prop:fidelity} reduces everything that follows to the study of the \emph{event} $A_\ML$ on the product probability space $(\FF_2^{2n}, \mu) = \prod_i (\FF_2^{2}, \mu_i)$: the optimal quantum decoder succeeds \emph{deterministically}, conditionally on the classical error pattern lying in $A_\ML$. This mirrors the classical situation, where maximum a posteriori decoders may be taken deterministic without loss of generality, and is exactly the property that general quantum codes lack.
\end{remark}

\section{Blowing up}
\label{sec:blowup}

For $A \subseteq \FF_2^{2n}$ and $r \geq 0$ let $\Gamma_r A := \{ e : d(e, A) \leq r \}$ denote the Hamming blow-up, with $d(e,A) = \min_{f\in A} \wt(e-f)$. The blowing-up lemma is due to Ahlswede, G\'acs and K\"orner~\cite{agk76} and Marton~\cite{marton86}; the explicit finite-blocklength form we use is not stated verbatim there, but follows in a few lines from the bounded-difference inequality~\cite{mcdiarmid89}, as we now show.

\begin{lemma}[Blowing-up]
\label{lem:blowup}
Let $\mu = \mu_1 \otimes \cdots \otimes \mu_n$ be a product probability measure on $\Omega = \prod_{i=1}^n \Omega_i$ with finite factors, equipped with the Hamming metric. If $\mu(A) \geq c > 0$, then for every $r \geq \sqrt{(n/2) \ln(1/c)}$,
\begin{align}
  \mu\bigl(\Gamma_r A\bigr) \;\geq\; 1 - \exp\!\left( - \frac{2}{n} \Bigl(r - \sqrt{(n/2)\ln(1/c)}\Bigr)^2 \right).
\end{align}
\end{lemma}

\begin{proof}
Let $f(x) := d(x, A) = \min_{a \in A}\wt(x-a)$ be the Hamming distance to $A$. Changing a single coordinate of $x$ changes $f(x)$ by at most $1$, so $f$ has the bounded-difference property with all constants equal to $1$, and McDiarmid's inequality~\cite{mcdiarmid89} applies: for every $t \geq 0$,
\begin{align}
  \mu\bigl(f \le \mathbb E f - t\bigr) \;\le\; e^{-2t^2/n}, \qquad
  \mu\bigl(f \geq \mathbb E f + t\bigr) \;\leq\; e^{-2t^2/n} .
\end{align}
We first bound the mean. As $\{f = 0\} = A$ has $\mu(A) \geq c$, the lower tail at $t = \mathbb E f$ gives $c \leq \mu(f \le 0) \leq e^{-2(\mathbb E f)^2/n}$, and solving for $\mathbb E f$,
\begin{align}
  \mathbb E f \;\leq\; \sqrt{(n/2)\ln(1/c)} \;=:\; \bar r .
\end{align}
Now $\Gamma_r A = \{ f \leq r \}$, so $1 - \mu(\Gamma_r A) = \mu(f > r)$. For $r \geq \bar r \geq \mathbb E f$ the value $t = r - \mathbb E f \geq 0$ is admissible in the upper tail, which gives
\begin{align}
  1 - \mu(\Gamma_r A) \;=\; \mu(f > r) \;\leq\; \mu\bigl(f \geq \mathbb E f + t\bigr) \;\leq\; e^{-2(r - \mathbb E f)^2/n} \;\leq\; e^{-2(r - \bar r)^2/n} ,
\end{align}
the final step using $r - \mathbb E f \geq r - \bar r \geq 0$. This is the claim.
\end{proof}

We will also need the volume of Hamming balls in the quaternary alphabet $\FF_2^2$: for $\delta \in (0, 3/4]$,
\begin{align}
  \bigl| \{ u \in \FF_2^{2n} : \wt(u) \leq \delta n \} \bigr|
  \;=\; \sum_{j \leq \delta n} \binom{n}{j} 3^j
  \;\leq\; 2^{\,n\, g(\delta)} ,
  \label{eq:volume}
\end{align}
the standard bound $\sum_{j\le \delta n}\binom nj (q-1)^j \le q^{n H_q(\delta)}$ for $\delta \le 1 - 1/q$ with $q = 4$ (see, e.g.,~\cite[Ch.~2]{ck11}), where $q^{H_q(\delta)} = 2^{h(\delta) + \delta\log(q-1)}$.

\section{The flagged channel}
\label{sec:flag}

The three lemmas of this section make no use of the stabilizer structure, and we state them for arbitrary isometric encoders; the generality costs nothing and will be needed in Sections~\ref{sec:main} and~\ref{sec:limits}. Fix an isometry $V: \mathbb C^D \to \cH_2^{\otimes n}$ with $D = 2^k$, a $\theta$-core $(A, \widetilde\cD)$ for it in the sense of Definition~\ref{def:core}, and $\delta \in (0, 3/4]$. Let
\begin{align}
  G_\delta := \{ u \in \FF_2^{2n} : \wt(u) \leq \delta n \} \cup \{\perp\},
  \qquad \log |G_\delta| \leq n\, g(\delta) + 1
\end{align}
by~\eqref{eq:volume}, and fix any function $u^\star : \FF_2^{2n} \to G_\delta$ such that $u^\star(e) \in \operatorname{arg\,min}\{\wt(u) : e - u \in A\}$ whenever $d(e, A) \leq \delta n$, and $u^\star(e) = \perp$ otherwise. Define the \emph{flagged channel}
\begin{align}
  \cM_\delta(\rho) := \sum_{e \in \FF_2^{2n}} \mu(e)\; \sigma_e \rho \sigma_e \otimes \proj{u^\star(e)}_{G} ,
  \label{eq:flagged}
\end{align}
a channel from $n$ qubits to $n$ qubits plus a classical register $G$ of size $|G_\delta|$. Tracing out $G$ recovers $\cN^{(n)}$ exactly. Note that $\cM_\delta$ depends on the code under scrutiny (through the core); this is unproblematic, because the two lemmas below and the converse bound of Lemma~\ref{lem:oneshot} hold for arbitrary channels.

\begin{lemma}[Fidelity boost]
\label{lem:boost}
There is a decoder $\cD'$ for $\cM_\delta$ (with the same encoder $V$) whose entanglement fidelity satisfies
\begin{align}
  F' \;\geq\; (1-\theta)\, \mu\bigl( \Gamma_{\delta n} A \bigr).
\end{align}
\end{lemma}

\begin{proof}
Let $\cD'$ measure $G$; on outcome $u \neq\, \perp$ apply $\sigma_u$ and then $\widetilde\cD$; on outcome $\perp$ output an arbitrary fixed state. If $d(e, A) \le \delta n$, the flag is $u = u^\star(e)$ and $\sigma_u \sigma_e = \eta\, \sigma_{e - u}$ with $e - u \in A$, so the corresponding branch is decoded with fidelity $f_{e-u}(V, \widetilde\cD) \geq 1 - \theta$. Hence $F' \geq (1-\theta) \sum_{e:\, u^\star(e)\ne \perp} \mu(e) = (1-\theta)\,\mu(\Gamma_{\delta n} A)$.
\end{proof}

\begin{lemma}[Entropic cost of the flag]
\label{lem:cost}
For every pure state $\phi_{R A^n}$, the state $\omega := (\id_R \otimes \cM_\delta)(\phi)$ satisfies
\begin{align}
  I(R \rangle B^n G)_\omega \;\leq\; I(R\rangle B^n)_{\omega'} + \log |G_\delta| \;\leq\; I(R\rangle B^n)_{\omega'} + n\, g(\delta) + 1 ,
\end{align}
where $\omega' := (\id \otimes \cN^{(n)})(\phi)$. In particular the first term equals $\Ic(V)$ when $\phi = \phi_V$, and is at most $Q^{(n)}$ for every $\phi$.
\end{lemma}

\begin{proof}
First, $\tr_G \omega = \omega'$: the flagged channel acts on each error branch as $\cN^{(n)}$ does, and merely appends the classical record $u^\star(e)$ in the register $G$; discarding an appended register recovers the branch, and hence the average, exactly.

Next, the chain rule
\begin{align}
  I(R\rangle B^n G)_\omega \;=\; I(R \rangle B^n)_{\omega'} + I(R : G \,|\, B^n)_\omega
  \label{eq:flagchain}
\end{align}
holds because, upon expanding the conditional mutual information as $H(RB^n) + H(B^nG) - H(B^n) - H(RB^nG)$, the right-hand side telescopes to $H(B^nG) - H(RB^nG) = I(R\rangle B^nG)$; entropies of marginals not involving $G$ agree between $\omega$ and $\omega'$.

It remains to bound the conditional mutual information by the size of the flag:
\begin{align}
  I(R:G|B^n) \;=\; H(G|B^n) - H(G|R B^n) \;\leq\; H(G|B^n) \;\leq\; H(G) \;\leq\; \log|G_\delta| \,.
  \label{eq:cmibound}
\end{align}
The last two inequalities are conditioning reduces entropy and the dimension bound for the classical register $G$. The first is $H(G | RB^n)_\omega \geq 0$, which is where classicality of the flag enters:\footnote{A \emph{quantum} register of the same dimension would only satisfy $H(G|RB^n) \geq -\log|G_\delta|$, doubling the cost in~\eqref{eq:cmibound}; the classical flag thus saves a factor of two in the exponent $g(\delta)$. Nothing below depends on this.} writing $\omega = \sum_g p_g\, \omega_g^{RB^n} \otimes \proj{g}$, one has $H(G|RB^n)_\omega = H(p) + \sum_g p_g H(\omega_g) - H\bigl(\sum_g p_g \omega_g\bigr) \geq 0$ by concavity of the entropy.
\end{proof}

\begin{lemma}[One-shot converse with error]
\label{lem:oneshot}
Let $\cT$ be any channel from $A^n$ to a system $B'$, let $V: \mathbb C^D \to \cH_{A^n}$ be an isometry, and let $\cD'$ be any decoder. Set $\sigma_{RL} := (\id \otimes \cD' \circ \cT \circ \cV)(\proj\Phi)$ and $f := \bra\Phi \sigma \ket\Phi$. Then, with $\phi := (\Idop \otimes V)\ket\Phi$,
\begin{align}
  \log D \;\leq\; I(R\rangle B')_{(\id\otimes\cT)(\proj\phi)} \;+\; (1-f)\log\bigl(D^2-1\bigr) \;+\; h(f) .
  \label{eq:oneshot}
\end{align}
In particular, if $f \geq 1-\varepsilon$ for some $\varepsilon \leq 1 - D^{-2}$, then
\begin{align}
  \log D \;\leq\; I(R\rangle B')_{(\id\otimes\cT)(\proj\phi)} + 2\varepsilon \log D + 1 .
  \label{eq:oneshoteps}
\end{align}
\end{lemma}

\begin{proof}
For $D = 1$ the reference and logical systems are trivial, $f = 1$, and both bounds reduce to $0 \leq 0$; assume henceforth $D \geq 2$. Let $\cU(\rho) := \int (\bar U \otimes U)\, \rho\, (\bar U \otimes U)^\dagger \,\mathrm{d}U$ be the twirl over Haar-random unitaries $U$ on $\cH_L$, with $\bar U$ the entrywise complex conjugate acting on $\cH_R$. By~\eqref{eq:ricochet}, $(\bar U\otimes U)\ket\Phi = (\bar U U^T \otimes \Idop)\ket\Phi = \ket\Phi$, so $\tilde\sigma := \cU(\sigma)$ has the same overlap $f$ with $\Phi$. Moreover $\tilde\sigma$ commutes with every $\bar U \otimes U$, and by Schur's lemma it is determined by its overlap with $\Phi$: the representation $U \mapsto \bar U \otimes U$ of the unitary group decomposes into the trivial representation, spanned by $\ket\Phi$, and an irreducible complement of dimension $D^2-1$, so any invariant state is a mixture of the two corresponding normalized projections,
\begin{align}
  \tilde\sigma \;=\; f \proj\Phi \;+\; \frac{1-f}{D^2-1}\bigl(\Idop_{RL} - \proj\Phi\bigr) .
\end{align}
Its spectrum gives $H(RL)_{\tilde\sigma} = h(f) + (1-f)\log(D^2-1)$, and its marginal on $L$ is maximally mixed, so $H(L)_{\tilde\sigma} = \log D$ and
\begin{align}
  I(R\rangle L)_{\tilde\sigma} \;=\; \log D - h(f) - (1-f)\log\bigl(D^2-1\bigr) .
\end{align}
The conditional entropy is invariant under local unitaries and concave (Section~\ref{sec:notation}), so
\begin{align}
  H(R|L)_{\tilde\sigma} \;\geq\; \int H(R|L)_{(\bar U\otimes U)\sigma(\bar U\otimes U)^\dagger}\,\mathrm{d}U \;=\; H(R|L)_\sigma \,,
\end{align}
that is, $I(R\rangle L)_\sigma \geq I(R\rangle L)_{\tilde\sigma}$. Data processing on the decoder then yields $I(R \rangle L)_\sigma \leq I(R\rangle B')_{(\id\otimes\cT)(\proj\phi)}$, and combining the displays gives~\eqref{eq:oneshot}. For~\eqref{eq:oneshoteps}, note that the map $f \mapsto h(f) + (1-f)\log(D^2-1)$ has derivative $\log\bigl(\frac{1-f}{f}\bigr) - \log(D^2-1)$, negative exactly for $f > D^{-2}$; so the right-hand side of~\eqref{eq:oneshot} is nonincreasing in $f$ on $[D^{-2},1]$ and we may replace $f$ by $1-\varepsilon$, then use $\log(D^2-1) \leq 2\log D$ and $h \leq 1$.
\end{proof}

\begin{remark}[Continuity of the conditional entropy]
\label{rem:twirl}
The twirl does two things at once: it makes the marginals of $\tilde\sigma$ and $\Phi$ on the conditioning system $L$ agree, and it shrinks the distance to the target, since the trace distance $T(\sigma,\Phi)$ can be as large as $\sqrt{1-f}$ whereas $T(\tilde\sigma,\Phi) = 1-f$ exactly. Both gains are best appreciated against the continuity bounds one would otherwise reach for, which have a long history.

Continuity of the von Neumann entropy in trace distance goes back to Fannes~\cite{fannes73}, with the sharp constant due to Audenaert~\cite{audenaert07}; for the \emph{conditional} entropy the standard tool is the Alicki--Fannes bound~\cite{alickifannes04} in Winter's tight form~\cite{winter16}. Following the sharp classical constant of Alhejji and Smith~\cite{alhejjismith20}, Wilde~\cite{wilde20} conjectured the sharp quantum bound $\epsilon\log(|A|^2-1) + h(\epsilon)$; this remains open in general, but Berta, Lami and Tomamichel~\cite{bertalamitomamichel25} proved it whenever the two states have the same marginal on the conditioning system, and the same bound was reached by different means by Audenaert et al.~\cite{audenaertetal25}.

After the twirl our situation is exactly that equal-marginal case, and moreover the extremal one, saturated by a maximally entangled state paired with an isotropic state---the pair appearing in the proof above. Invoking~\cite{bertalamitomamichel25} therefore returns~\eqref{eq:oneshot} verbatim, and the explicit computation in the proof may be read as a self-contained substitute for it. Passing through the trace distance instead, via~\cite{fvdg99} and~\cite{alickifannes04,winter16}, would give only $\log D \leq I(R\rangle B') + 2\sqrt{1-f}\,\log D + 2$; the linear rather than square-root dependence on the error is what allows Theorem~\ref{thm:cores} to tolerate a core imperfection of order $\gamma$ rather than $\gamma^2$.
\end{remark}

\section{Main results for qubit Pauli channels}
\label{sec:main}

Nothing in Sections~\ref{sec:blowup} and~\ref{sec:flag} used the stabilizer structure: the machinery needs only a core. We therefore state the master inequality in that generality first and specialize afterwards, the specialization being supplied by Proposition~\ref{prop:fidelity}, which hands us a $0$-core whose mass is the optimal fidelity.

\begin{theorem}[Master inequality]
\label{thm:coremaster}
Let $V$ be an arbitrary isometric encoder of $k$ qubits into $n$ uses of a product Pauli channel $\cN^{(n)}$ as in~\eqref{eq:productchannel}, and let $(A, \widetilde\cD)$ be a $\theta$-core for $V$. Then for every $\delta \in (0,3/4]$,
\begin{align}
  k \;\leq\; \Ic(V) + n\, g(\delta) + 2 + 2k\bigl(\theta + 1 - \mu(\Gamma_{\delta n} A)\bigr) ,
  \label{eq:coremaster}
\end{align}
and, whenever $\tau := \sqrt{\ln(1/\mu(A))/(2n)} \leq \delta$,
\begin{align}
  1 - \mu\bigl(\Gamma_{\delta n} A\bigr) \;\leq\; \exp\bigl(-2n(\delta - \tau)^2\bigr) .
  \label{eq:coreblowup}
\end{align}
\end{theorem}

\begin{proof}
Instantiate the flagged channel of Section~\ref{sec:flag} with the core $(A, \widetilde\cD)$ and the given $\delta$, obtaining $\cM_\delta$ together with a classical flag register $G$ of size $\log|G_\delta| \leq n\,g(\delta) + 1$.

By Lemma~\ref{lem:boost} there is a decoder $\cD'$ for $\cM_\delta$, using the \emph{same} encoder $\cV$, whose entanglement fidelity satisfies $F' \geq (1-\theta)\,\mu(\Gamma_{\delta n} A)$. Its error is accordingly
\begin{align}
  \varepsilon \;:=\; 1 - F' \;\leq\; 1 - (1-\theta)\,\mu(\Gamma_{\delta n} A) \;\leq\; \theta + \bigl(1 - \mu(\Gamma_{\delta n} A)\bigr) ,
  \label{eq:coreeps}
\end{align}
the last inequality using $\mu(\Gamma_{\delta n} A) \leq 1$. The flagged scheme $(\cV, \cD')$ thus operates at fidelity $1-\varepsilon$, in the regime of the one-shot converse.

The assertion is immediate for $k = 0$, when~\eqref{eq:coremaster} reads $0 \leq \Ic(V) + n\,g(\delta) + 2$ and holds since $\Ic(V) \geq -k = 0$ and $g(\delta) \geq 0$; assume henceforth $k \geq 1$.

Assume first that $\varepsilon \leq 1 - D^{-2}$, so that the one-shot converse~\eqref{eq:oneshoteps} applies to the flagged channel $\cT = \cM_\delta$ with encoder $\cV$, decoder $\cD'$, $\log D = k$ and $f = 1-\varepsilon$. Chaining it with the flag-cost estimate of Lemma~\ref{lem:cost} at the input $\phi_V$, and then with~\eqref{eq:coreeps},
\begin{align}
  k
  &\;\leq\; I(R\rangle B^n G)_{(\id\otimes\cM_\delta)(\phi_V)} + 2 k \varepsilon + 1 \nonumber\\
  &\;\leq\; \Ic(V) + n\,g(\delta) + 2 + 2k\varepsilon \nonumber\\
  &\;\leq\; \Ic(V) + n\,g(\delta) + 2 + 2k\bigl(\theta + 1 - \mu(\Gamma_{\delta n} A)\bigr) ,
\end{align}
where the first inequality is~\eqref{eq:oneshoteps}, the second is Lemma~\ref{lem:cost} (bounding the flagged coherent information by $\Ic(V) + \log|G_\delta| \leq \Ic(V) + n\,g(\delta) + 1$), and the third is~\eqref{eq:coreeps}. This is~\eqref{eq:coremaster}.

If instead $\varepsilon > 1 - D^{-2}$, then~\eqref{eq:oneshoteps} does not apply, but~\eqref{eq:coremaster} still holds, and trivially. Indeed, by~\eqref{eq:coreeps} and $2k\,2^{-2k} \leq 1$,
\begin{align}
  2k\bigl(\theta + 1 - \mu(\Gamma_{\delta n} A)\bigr) \;\geq\; 2k\varepsilon \;>\; 2k\bigl(1 - 2^{-2k}\bigr) \;\geq\; 2k - 1 ,
\end{align}
so the right-hand side of~\eqref{eq:coremaster}, bounded below using $\Ic(V) \geq -k$ and $n\,g(\delta) \geq 0$, exceeds $-k + 2 + (2k - 1) = k + 1 > k$.

Finally,~\eqref{eq:coreblowup} is Lemma~\ref{lem:blowup} applied to $A$ with $c = \mu(A)$ and radius $r = \delta n$. Its hypothesis $r \geq \sqrt{(n/2)\ln(1/c)}$ is exactly $\delta \geq \tau$, and since $\sqrt{(n/2)\ln(1/\mu(A))} = n\tau$ its conclusion reads
\begin{align}
  1 - \mu(\Gamma_{\delta n} A) \;\leq\; \exp\!\Bigl(-\tfrac{2}{n}\,(\delta n - n\tau)^2\Bigr) \;=\; \exp\bigl(-2n(\delta-\tau)^2\bigr) ,
\end{align}
which is~\eqref{eq:coreblowup}.
\end{proof}

The rate is thus measured against the coherent information of the code's \emph{own} input state, not against the maximum $Q^{(n)}$ over all inputs; since $\Ic(V) \le Q^{(n)}$ this is the sharper statement, and it is what makes the corollaries below possible. The overshoot rate is accordingly $\gamma := (k - \Ic(V))/n$, which by $k \le n$ and $\Ic(V) \ge -k$ lies in $[0,2]$, inside the range $(0,2]$ on which $\gamma \mapsto g^{-1}(\gamma/2)$ is well defined by~\eqref{eq:gdef}. The hypotheses below are monotone in $\gamma$, so restricting to $\gamma \in (0,1]$ costs nothing. Throughout the rest of the paper we write
\begin{align}
  \delta_\gamma := g^{-1}(\gamma/2) \in (0, 3/4], \qquad E(\gamma) := \frac{\delta_\gamma^2}{2} \,.
  \label{eq:expdef}
\end{align}

\begin{theorem}[No heavy cores above the coherent information]
\label{thm:cores}
In the setting of Theorem~\ref{thm:coremaster}, suppose $k \geq \Ic(V) + \gamma n$ with $\gamma \in (0,1]$, and let $\theta \leq \gamma/16$. Then
\begin{align}
  \mu(A) \;\leq\; e^{-n E(\gamma)} \qquad \text{for all } n \geq n_1(\gamma) := \Bigl\lceil \max\Bigl\{ \frac{16}{\gamma},\; \frac{2}{\delta_\gamma^2}\ln\frac{16}{\gamma} \Bigr\} \Bigr\rceil .
\end{align}
No code whatsoever admits a near-deterministic core of mass $e^{-o(n)}$ once its rate exceeds the coherent information $\Ic(V)$ of its own input state; in particular, the conclusion applies at every rate above the $n$-letter maximum $Q^{(n)}$.
\end{theorem}

\begin{proof}
Suppose $\mu(A) > e^{-n\delta_\gamma^2/2}$. Then $\tau < \delta_\gamma/2$, so~\eqref{eq:coreblowup} gives $1 - \mu(\Gamma_{\delta_\gamma n} A) \le e^{-n\delta_\gamma^2/2}$, and using $k \le n$, $\theta \le \gamma/16$ and $k - \Ic(V) \geq \gamma n$, inequality~\eqref{eq:coremaster} with $\delta = \delta_\gamma$ yields
\begin{align}
  \gamma n \;\le\; \frac{\gamma n}{2} + 2 + \frac{\gamma n}{8} + 2n\, e^{-n\delta_\gamma^2/2} .
\end{align}
For $n \geq 16/\gamma$ we have $2 \le \gamma n / 8$, and for $n \ge (2/\delta_\gamma^2)\ln(16/\gamma)$ we have $2n\, e^{-n\delta_\gamma^2/2} \le \gamma n/8$; the right-hand side is then at most $7\gamma n/8 < \gamma n$, a contradiction.
\end{proof}

\begin{remark}[Fidelity versus core mass]
\label{rem:coresfid}
For a stabilizer encoder the fidelity of any scheme is at most---and, for the optimal decoder, exactly---the largest mass of a $0$-core: this is Proposition~\ref{prop:fidelity}, together with $F \geq \mu(A)$ for any $0$-core $(A, \widetilde\cD)$. Bounding cores therefore bounds fidelity, which is how everything below follows from Theorem~\ref{thm:coremaster}. For general encoders no such identity holds: Proposition~\ref{prop:nogo} exhibits schemes of constant fidelity all of whose near-deterministic cores have exponentially small mass. Theorem~\ref{thm:cores} accordingly bounds cores rather than fidelity, and the strong converse for all codes would follow from the extraction statement of Problem~\ref{prob:core}.
\end{remark}

\begin{theorem}[Master inequality for stabilizer codes]
\label{thm:master}
Let $\cN^{(n)}$ be a product of qubit Pauli channels as in~\eqref{eq:productchannel} and let $(S, V, \cD)$ be an $(n,k)$ stabilizer coding scheme with optimal fidelity $F^\star = \mu(A_\ML)$ as in Proposition~\ref{prop:fidelity}. Then for every $\delta \in (0, 3/4]$,
\begin{align}
  k \;\leq\; \Ic(V) + n\, g(\delta) + 2 + 2k \bigl( 1 - \mu\bigl(\Gamma_{\delta n} A_\ML\bigr)\bigr)\,,
  \label{eq:master}
\end{align}
and, whenever $\tau := \sqrt{\ln(1/F^\star)/(2n)} \leq \delta$,
\begin{align}
  1 - \mu\bigl(\Gamma_{\delta n} A_\ML\bigr) \;\leq\; \exp\bigl(-2 n\, (\delta - \tau)^2\bigr) .
  \label{eq:masterblowup}
\end{align}
\end{theorem}

\begin{proof}
By Proposition~\ref{prop:fidelity}, the pair $(A_\ML, \cD_\ML)$ is a $0$-core for $V$ of mass $\mu(A_\ML) = F^\star$. Apply Theorem~\ref{thm:coremaster} with $\theta = 0$ and $A = A_\ML$; the two displays are then~\eqref{eq:coremaster} and~\eqref{eq:coreblowup}.
\end{proof}

\begin{theorem}[Exponential decay of fidelity]
\label{thm:decay}
In the setting of Theorem~\ref{thm:master}, suppose $k \geq \Ic(V) + \gamma n$ for some $\gamma \in (0, 1]$, with $\delta_\gamma$ and $E(\gamma)$ as in~\eqref{eq:expdef}. Then the fidelity of the scheme satisfies
\begin{align}
  F \;\leq\; F^\star \;\leq\; e^{-n E(\gamma)} \qquad \text{for all } n \geq n_0(\gamma) := \Bigl\lceil \max\Bigl\{ \frac{8}{\gamma},\; \frac{2}{\delta_\gamma^2} \ln \frac{8}{\gamma} \Bigr\} \Bigr\rceil ,
\end{align}
and $F \leq C(\gamma)\, e^{-nE(\gamma)}$ for all $n$, with $C(\gamma) := e^{\,n_0(\gamma) E(\gamma)}$. The average and worst-case pure-state fidelities obey the same bound: $F_{\mathrm{av}} \le F + 1/D$, and here $1/D = 2^{-k} \le e^{-nE(\gamma)}$---since $k \ge \gamma n/2$ (as $\Ic(V) \ge -k$) and $E(\gamma) = \delta_\gamma^2/2 \le \gamma^2/8 \le \tfrac12\gamma\ln 2$ (as $\delta_\gamma \le \gamma/2$, which follows from $g(\delta) \ge \delta$)---so both are at most $(C(\gamma)+1)\,e^{-nE(\gamma)}$.
\end{theorem}

\begin{proof}
Suppose toward a contradiction that $F^\star > e^{-n \delta_\gamma^2/2}$ and $n \geq n_0(\gamma)$. Then $2n\tau^2 = \ln(1/F^\star) < n \delta_\gamma^2/2$, that is, $\tau < \delta_\gamma/2$, so~\eqref{eq:masterblowup} yields
\begin{align}
  1 - \mu\bigl(\Gamma_{\delta_\gamma n} A_\ML\bigr) \;\leq\; e^{-2n(\delta_\gamma - \tau)^2} \;<\; e^{-n\delta_\gamma^2/2} .
\end{align}
Substituting into~\eqref{eq:master} with $\delta = \delta_\gamma$, and using $k \le n$, $g(\delta_\gamma) = \gamma/2$ and $k - \Ic(V) \geq \gamma n$:
\begin{align}
  \gamma n \;<\; \frac{\gamma n}{2} + 2 + 2 n\, e^{-n \delta_\gamma^2/2} .
\end{align}
For $n \geq 8/\gamma$ we have $2 \leq \gamma n/4$, and for $n \geq (2/\delta_\gamma^2)\ln(8/\gamma)$ we have $2n\, e^{-n\delta_\gamma^2/2} \leq \gamma n/4$; the right-hand side above is then at most $\gamma n$, so the strict inequality gives $\gamma n < \gamma n$, a contradiction. Hence $F^\star \leq e^{-n\delta_\gamma^2/2}$ for $n \geq n_0(\gamma)$; for smaller $n$ use $F^\star \leq 1 \leq C(\gamma) e^{-nE(\gamma)}$.
\end{proof}

We now record the consequences announced in the introduction.

\begin{corollary}[Strong converse at zero rate for antidegradable Pauli channels]
\label{cor:antideg}
Let each $\cN_i$ be an antidegradable Pauli channel. Then every $(n,k)$ stabilizer coding scheme with $k \geq \gamma n$, $\gamma \in (0,1]$, has $F \leq C(\gamma) e^{-nE(\gamma)}$, with $C(\gamma)$ and $E(\gamma)$ as in Theorem~\ref{thm:decay}. Consequently the $\varepsilon$-quantum capacity~\eqref{eq:epscap} within the class of stabilizer coding schemes satisfies $Q_\varepsilon = 0$ for every $\varepsilon \in (0,1)$; in the memoryless case $Q_\varepsilon = 0 = Q(\cN)$, so the strong converse property holds here in the strict sense. This applies to the qubit depolarizing channel with error probability $p \in [1/4, 3/4]$.
\end{corollary}

\begin{proof}
We first check that $Q^{(n)} = 0$. Complementary channels compose over tensor products, so if $\cN_i = \Theta_i \circ \cN_i^c$ for channels $\Theta_i$, then $\cN^{(n)} = (\Theta_1 \otimes \cdots \otimes \Theta_n) \circ (\cN^{(n)})^c$ and the product is again antidegradable. Let $\phi_{RA^n}$ be a pure input, with dilated output pure on $R B^n E^n$. Antidegradability means $\omega_{RB^n} = (\id_R \otimes \Theta)(\omega_{RE^n})$, so data processing gives $I(R\rangle B^n) \leq I(R\rangle E^n)$, while purity gives the duality $I(R\rangle E^n) = -I(R\rangle B^n)$. Hence $I(R\rangle B^n) \leq 0$ for every input and $Q^{(n)} \leq 0$; pure inputs attain $0$, so $Q^{(n)} = 0$. In particular $\Ic(V) \leq 0$ for every encoder, so Theorem~\ref{thm:decay} applies to any $k \geq \gamma n$ and gives the stated bound.

For the $\varepsilon$-capacity, fix $\varepsilon \in (0,1)$ and $\gamma \in (0,1]$. Since $C(\gamma) e^{-nE(\gamma)} < 1 - \varepsilon$ for all sufficiently large $n$, no $(n,k)$ stabilizer scheme with $k \geq \gamma n$ is $\varepsilon$-good once $n$ is that large, so the supremum in~\eqref{eq:epscap} is below $\gamma n$ and $Q_\varepsilon \leq \gamma$. Letting $\gamma \downarrow 0$ gives $Q_\varepsilon = 0$.
\end{proof}

\begin{corollary}[The strong converse property within the stabilizer class]
\label{cor:iid}
Let $\cN_i = \cN$ for all $i$ (memoryless Pauli channel) and let $Q_{\mathrm{stab}}(\cN)$ be as in~\eqref{eq:qstab}. Then every sequence of $(n,k_n)$ stabilizer coding schemes with $k_n \geq n\bigl(Q_{\mathrm{stab}}(\cN) + \gamma\bigr)$, $\gamma \in (0,1]$, satisfies $F_n \leq C(\gamma)e^{-nE(\gamma)}$, with $C(\gamma)$ and $E(\gamma)$ as in Theorem~\ref{thm:decay}. Consequently the $\varepsilon$-quantum capacity~\eqref{eq:epscap} within the class of stabilizer coding schemes satisfies
\begin{align}
  Q_\varepsilon \;=\; Q_{\mathrm{stab}}(\cN) \qquad \text{for every } \varepsilon \in (0,1) :
\end{align}
the class has the strong converse property, and its $\varepsilon$-capacity is Hamada's conditional capacity. Since $Q_{\mathrm{stab}}(\cN) \leq Q(\cN)$, the quantum capacity is in particular a strong converse rate for the class.
\end{corollary}

\begin{proof}
Write $Q^{(n)}_{\mathrm{stab}} := \max_{\cC \in \mathsf{S}_n} \Ic(\Pi_\cC/\tr\Pi_\cC, \cN^{\otimes n})$, so that $\Ic(V) \leq Q^{(n)}_{\mathrm{stab}}$ for every $(n,k)$ stabilizer scheme, by~\eqref{eq:Icode} and Proposition~\ref{prop:fidelity}. Pasting two stabilizer codes side by side shows $Q_{\mathrm{stab}}^{(n+m)} \geq Q^{(n)}_{\mathrm{stab}} + Q^{(m)}_{\mathrm{stab}}$~\cite[Eq.~(58)]{hamada05}, so Fekete's lemma gives $Q^{(n)}_{\mathrm{stab}} \leq n\, Q_{\mathrm{stab}}(\cN)$ for every $n$. A scheme with $k_n \geq n(Q_{\mathrm{stab}} + \gamma)$ therefore has $k_n \geq \Ic(V) + \gamma n$, and Theorem~\ref{thm:decay} applies.

For the $\varepsilon$-capacity, fix $\varepsilon \in (0,1)$. Given $\gamma \in (0,1]$, the displayed bound falls below $1-\varepsilon$ for large $n$, so no $\varepsilon$-good stabilizer scheme has rate $Q_{\mathrm{stab}} + \gamma$ or more once $n$ is large; hence $Q_\varepsilon \leq Q_{\mathrm{stab}} + \gamma$, and $\gamma \downarrow 0$ gives $Q_\varepsilon \leq Q_{\mathrm{stab}}$. The reverse inequality is Hamada's achievability~\cite{hamada05}: concatenated stabilizer codes attain the rate $Q_{\mathrm{stab}}(\cN)$ with fidelity tending to one, and such schemes are $\varepsilon$-good for every $\varepsilon$ once $n$ is large.
\end{proof}

\begin{remark}[What Corollary~\ref{cor:iid} does and does not settle]
\label{rem:notequal}
Hamada~\cite{hamada05} determined the class's capacity at vanishing error, $\lim_{\varepsilon\to 0} Q_\varepsilon = Q_{\mathrm{stab}}$; Corollary~\ref{cor:iid} adds that tolerating a constant error does not raise it, and that the fidelity in fact collapses exponentially above it. What remains open is whether the restriction to stabilizer codes costs anything at all, that is, whether $Q_{\mathrm{stab}}(\cN) = Q(\cN)$; this is conjectured in~\cite{hamada05}. Should it hold, $Q_\varepsilon = Q(\cN)$ for every $\varepsilon < 1$ and the strong converse for Pauli channels would follow outright---but only for stabilizer codes, so that Problem~\ref{prob:core} would still stand between this and the general statement.
\end{remark}

\begin{corollary}[Shared randomness]
\label{cor:randomized}
Consider a scheme given by an ensemble $\bigl(S^{(j)}, V^{(j)}, \cD^{(j)}\bigr)_j$ of $(n,k)$ stabilizer coding schemes, with $j$ drawn from a distribution $(p_j)_j$ available to both encoder and decoder, and let $A_\ML^{(j)}$ be the maximum-likelihood event~\eqref{eq:AML} of the $j$-th member. Then the average entanglement fidelity $F = \sum_j p_j F_j$ satisfies, for every $\delta \in (0,3/4]$,
\begin{align}
  k \;\leq\; \sum_j p_j \Ic\bigl(V^{(j)}\bigr) + n\, g(\delta) + 2 + 2k \sum_j p_j \Bigl( 1 - \mu\bigl(\Gamma_{\delta n} A_\ML^{(j)}\bigr) \Bigr) ,
  \label{eq:mastermixed}
\end{align}
and, if $k \geq \Ic(V^{(j)}) + \gamma n$ for every $j$ and some $\gamma \in (0,1]$, it satisfies the bound $F \leq C(\gamma)\, e^{-nE(\gamma)}$ of Theorem~\ref{thm:decay}.
\end{corollary}

\begin{proof}
The entanglement fidelity is affine in the coding scheme, so $F = \sum_j p_j F_j$ with $F_j$ the fidelity of the $j$-th member. Each member is an $(n,k)$ stabilizer coding scheme, so Theorem~\ref{thm:master} applies to it and yields~\eqref{eq:master} with $A_\ML^{(j)}$ and $\Ic(V^{(j)})$ in place of $A_\ML$ and $\Ic(V)$; averaging these inequalities with weights $p_j$ gives~\eqref{eq:mastermixed}. Likewise Theorem~\ref{thm:decay} gives $F_j \leq C(\gamma)e^{-nE(\gamma)}$ for every $j$, a bound uniform in $j$, so the average obeys it too.
\end{proof}

Note that the ensemble is not itself a stabilizer coding scheme: there is no single stabilizer group, hence no single maximum-likelihood event and no exact fidelity formula as in Proposition~\ref{prop:fidelity}. This is why the master inequality survives only in the averaged form~\eqref{eq:mastermixed}, one blow-up term per member, whereas the exponential bound, being a uniform constant over the class, survives verbatim. It is the latter that we use.

The exponential bound thus extends to \emph{incoherent} shared-randomness mixtures and, up to a counting factor, to \emph{coherent} superpositions of codes with non-overlapping code spaces, even though such ensembles are not themselves stabilizer coding schemes.

\begin{corollary}[Superpositions of orthogonal stabilizer codes]
\label{cor:superposition}
Let $V = \sum_{i=1}^{B} z_i V_i$, with $\sum_i |z_i|^2 = 1$, where every $V_i$ is an isometric encoder onto the code space of an $(n,k)$ stabilizer code (the stabilizer groups may differ) and the ranges of the $V_i$ are pairwise orthogonal, so that $V$ is an isometry. If $k \geq \Ic(V_i) + \gamma n$ for every $i$ and some $\gamma \in (0,1]$---in particular whenever $k \geq Q^{(n)} + \gamma n$---then for every decoder $\cD$,
\begin{align}
  F(V, \cD) \;\leq\; B\, C(\gamma)\, e^{-n E(\gamma)} ,
\end{align}
with $C(\gamma)$ and $E(\gamma)$ as in Theorem~\ref{thm:decay}. The strong converse therefore extends to coherent superpositions of subexponentially many, $B \leq e^{o(n)}$, mutually orthogonal stabilizer codes.
\end{corollary}

\begin{proof}
Fix a decoder $\cD$ with Kraus operators $\{M_a\}$ and an error pattern $e$, and write $\psi_e := (\Idop\otimes\sigma_e V)\ket\Phi = \sum_i z_i \psi_e^{(i)}$ with $\psi_e^{(i)} := (\Idop\otimes\sigma_e V_i)\ket\Phi$. By the triangle inequality in $\ell^2$ and the Cauchy--Schwarz inequality,
\begin{align}
  f_e(V,\cD) \;=\; \Bigl\| \bigl(\langle\Phi|(\Idop\otimes M_a)\psi_e\rangle\bigr)_a \Bigr\|_2^2
  \;\leq\; \Bigl( \sum_i |z_i| \sqrt{f_e(V_i,\cD)} \Bigr)^{\!2}
  \;\leq\; \sum_i f_e(V_i, \cD) \,.
\end{align}
Averaging over $e$ gives $F(V, \cD) \leq \sum_i F(V_i, \cD)$. Each $(V_i, \cD)$ is a stabilizer coding scheme at rate $k$ with $k \geq \Ic(V_i) + \gamma n$, so $F(V_i, \cD) \leq F_i^\star \leq C(\gamma)\, e^{-nE(\gamma)}$ by Proposition~\ref{prop:fidelity} and Theorem~\ref{thm:decay}.
\end{proof}

\begin{remark}[Robustness]
\label{rem:robust}
Suppose a (possibly non-stabilizer) encoder produces a state $\tau_{RA^n}$ with $T(\tau, \cV(\proj\Phi)) \leq \zeta$, for some stabilizer encoder $\cV$ in the class above. Then the fidelity of the scheme exceeds that of the corresponding stabilizer scheme by at most $\zeta$, since $|\!\bra\Phi (\id\otimes\cD\cN^{(n)})(\tau - \cV(\proj\Phi))\ket\Phi\!| \leq \zeta$ by contractivity. Hence any $\zeta \to 0$ preserves the qualitative strong converse, while the exponential bound is retained with the same exponent, $F \leq (C(\gamma)+1)\,e^{-nE(\gamma)}$, as long as $\zeta \leq e^{-cn}$ with $c \geq E(\gamma)$: codes that close to the stabilizer class inherit the full conclusion.
\end{remark}

\begin{remark}[Interpretation]
The stabilizer constructions historically used to establish achievable rates and superadditivity for Pauli channels lie in the class~\cite{dss98,smithsmolin07}, so in the memoryless case the converse of Theorem~\ref{thm:decay} matches them exactly. It need not cover every recent construction: the symmetric-subspace optimizers of~\cite{agarwaletal26} that currently give the best depolarizing threshold are permutation-invariant states not presented as full stabilizer-sector projections, and so are not known to lie in the class. We also stress that Theorem~\ref{thm:master} makes no memorylessness assumption: arbitrary product Pauli noise, for instance time-varying noise levels, is covered.
\end{remark}

\medskip
We close this section with the formulation promised in the introduction, which decouples the argument from the stabilizer structure entirely.

\begin{remark}[When blowing up beats the trivial bound]
\label{rem:deflation}
It is instructive to run the machinery on an arbitrary scheme $(V, \cD)$ of fidelity $F$, with no sharpness input. For $\beta \in (0,F)$, Markov's inequality applied to $F = \sum_e \mu(e) f_e(V,\cD)$ shows that $A_\beta := \{e : f_e(V,\cD) \geq \beta\}$ has $\mu(A_\beta) \geq (F-\beta)/(1-\beta)$, so $(A_\beta, \cD)$ is a $(1-\beta)$-core of constant mass. Choosing $\delta = n^{-1/4}$ in~\eqref{eq:coremaster} makes both $n\,g(\delta) = O(n^{3/4}\log n)$ and the blow-up error vanish relative to $n$ (here $\tau = O(n^{-1/2})$), and letting $\beta \uparrow F$ yields
\begin{align}
  k \bigl( 1 - 2(1 - F) - o(1) \bigr) \;\leq\; \Ic(V) + o(n) ,
\end{align}
which is precisely the quantitative weak converse obtained by applying Lemma~\ref{lem:oneshot} to $(V,\cD)$ directly, with no blowing up at all. The method gains over this trivial bound exactly when decoding success \emph{concentrates}---when cores exist with $\theta \to 0$ whose mass is not exponentially smaller than the fidelity. Proposition~\ref{prop:fidelity} asserts that stabilizer codes concentrate perfectly.
\end{remark}

\section{The limits of decoder-side sharpening}
\label{sec:limits}
\label{sec:restriction}

How restrictive is the stabilizer assumption, and can it be removed by a better decoder analysis? The exponential converse extends beyond the stabilizer class to shared-randomness mixtures, orthogonal superpositions and exponentially small perturbations (Corollaries~\ref{cor:randomized} and~\ref{cor:superposition}, Remark~\ref{rem:robust}), so by Theorem~\ref{thm:cores} any violating sequence must maintain its fidelity by coherent hedging across error patterns rather than by heavy near-deterministic cores. Classically no such hedging exists---maximum a posteriori decoders may be taken deterministic, and blowing up then applies to \emph{all} codes~\cite{agk76}---but quantumly it can be \emph{forced}, through a trade-off expressing that one decoder output cannot overlap well with many orthogonal targets.

\begin{lemma}[Orthogonal-target trade-off]
\label{lem:rigidity}
Let $\Phi$ be maximally entangled on $\cH_R \otimes \cH_L$ with $\dim\cH_R = D$, let $\psi \in \cH_R \otimes \cH_B$ be a unit vector, and let $A_1, \dots, A_m$ be unitaries on $\cH_R$ that are pairwise orthogonal in the trace inner product, $\tr(A_w^\dagger A_{w'}) = 0$ for $w \neq w'$. Then for every channel $\widetilde\cD$ from $B$ to $L$,
\begin{align}
 \sum_{w=1}^{m} \bra\Phi\bigl(\id\otimes\widetilde\cD\bigr)\Bigl(\proj{(A_w \otimes \Idop_B)\psi}\Bigr)\ket\Phi \;\leq\; 1 .
\end{align}
\end{lemma}

\begin{proof}
Set $\sigma := (\id\otimes\widetilde\cD)(\proj\psi)$ and $\ket{\Phi_w} := (A_w^\dagger\otimes\Idop)\ket\Phi$. Since $A_w$ acts on $R$ alone, it commutes with $\id\otimes\widetilde\cD$, so the $w$-th summand equals $\langle\Phi_w|\sigma|\Phi_w\rangle$. By~\eqref{eq:ricochet}, $\langle\Phi_w|\Phi_{w'}\rangle = \tr(A_w A_{w'}^\dagger)/D = 0$ for $w \neq w'$, and $\|\Phi_w\| = 1$, so the vectors $\Phi_w$ are orthonormal and $\sum_w \langle\Phi_w|\sigma|\Phi_w\rangle \leq \tr \sigma = 1$.
\end{proof}

The lemma combines the orthogonality of unitarily displaced maximally entangled states, familiar from dense coding and quantum authentication~\cite{bcgst02}, with Bessel's inequality; related overlap arguments occur in quantitative approximate error correction~\cite{benyoreshkov10}. New is only the use to which it is put in the proposition below.

\begin{proposition}[Fidelity without sharp cores]
\label{prop:nogo}
Let $\cN$ be the memoryless qubit depolarizing channel with error probability $p \in (0,3/4)$ such that the hashing rate $R_h := 1 - h(p) - p\log 3$ is positive, and fix $0 < R < R_h$, $c \in (0,1)$ and $\theta \in (0, 1-c)$. Then for all sufficiently large $n$ there is an $(n,k)$ coding scheme with isometric encoder and $k = \lceil R(n-1)\rceil$ whose entanglement fidelity satisfies
\begin{align}
 F \;\geq\; c\,\Bigl(1 - \frac{2p}{3}\Bigr)(1 - \varepsilon_n), \qquad \varepsilon_n \to 0 ,
\end{align}
while \emph{every} $\theta$-core $(A, \widetilde\cD)$ for its encoder has mass
\begin{align}
  \mu(A) \;\leq\; \frac1v\, \mu_{\max}^{\,k},
  \qquad v := \frac{\bigl(\sqrt{1-\theta}-\sqrt c\,\bigr)^2}{1-c} > 0, \qquad \mu_{\max} := \max\{1-p,\, p/3\} < 1 .
\end{align}
In particular, for every pair of constants $(c,\theta)$ with $\theta < 1-c$, a fidelity bounded below by the constant $c\,(1-2p/3)$, up to a vanishing correction, coexists with the exponential lightness of all $\theta$-cores.
\end{proposition}

\begin{proof}[Proof sketch]
A selector qubit couples, with amplitudes $\sqrt c$ and $\sqrt{1-c}$, a full-sector stabilizer encoder of the same rate to a junk encoder with scrambled logical content; measuring the selector and decoding the good branch gives the fidelity bound. For the core bound, a pattern in a $\theta$-core (with $\theta<1-c$) satisfies $f_e \ge 1-\theta > c$, which forces a nonzero contribution from the junk branch; Lemma~\ref{lem:rigidity} then limits the number of junk-logical classes that can occur within each syndrome to at most $1/v$, and counting shows each contributes an exponentially small share of the measure. See Appendix~\ref{app:nogo}.
\end{proof}

The construction locates the obstruction precisely: a same-rate full-sector stabilizer code is visibly contained in the encoder with constant amplitude, and what fails, over all decoders, is decoder-side access to it. It does not threaten the strong converse---it lives below capacity, and Theorem~\ref{thm:cores} forbids a heavy sharp component from retaining essentially the same rate above capacity---but it shows that any route to the general statement through the present method must extract sharpness on the \emph{encoder} side. This is the corrected form of the open problem.

\begin{problem}[Sharp-code extraction]
\label{prob:core}
Prove or refute: for every $c, \theta \in (0,1)$ there exist $\lambda(c,\theta) > 0$ and $\nu_n = o(n)$ such that whenever an $(n,k)$ coding scheme with isometric encoder over a product Pauli channel has entanglement fidelity at least $c$, there exists an $(n, k')$ coding scheme with isometric encoder, $k' \geq k - \nu_n$, admitting a $\theta$-core of mass at least $\lambda(c,\theta)\, e^{-\nu_n}$.
\end{problem}

An affirmative answer would give the strong converse for the quantum capacity of every memoryless Pauli channel with respect to \emph{all} isometrically encoded codes: extraction would turn a scheme of rate $Q(\cN)+\gamma$ and fidelity at least $c$ into one of rate $Q(\cN)+\gamma/2$ carrying a heavy $\theta$-core, contradicting Theorem~\ref{thm:cores}. The decomposable cases are easy---Corollary~\ref{cor:superposition} handles encoders that split into orthogonal stabilizer blocks---but a general argument is genuinely hard: the two-element case of Lemma~\ref{lem:rigidity} is tight at conditional fidelity $\tfrac12$, so only a many-branch argument can work, and the stability theorem for approximate homomorphisms that this seems to require~\cite{gowershatami17} we could not establish in the form needed. We nonetheless believe Problem~\ref{prob:core} has an affirmative answer.

Finally, genuine quantum blowing-up lemmas do not shortcut this. Osborne and Winter~\cite{osbornewinter09} fatten a \emph{subspace} by its low-weight Pauli displacements, and Talagrand--KKL-type inequalities have reached the quantum hypercube~\cite{rouzewirthzhang24,changli26}, but concentration enters our argument through Lemma~\ref{lem:boost}, whose decoder must act on the certificate by applying $\sigma_{u^\star(e)}$; the subspace join forgets which displacement was used, and recovering it from a superposition of displaced branches meets exactly the hedging above. The barrier for general codes is not the absence of quantum concentration but of a sharp object to concentrate.

\section{Discussion}
\label{sec:discussion}

\subsection{Two-way classical assistance must be excluded}
\label{sec:feedback}

\begin{remark}
\label{rem:feedback}
No bound of the form of Theorem~\ref{thm:decay} can hold in the presence of free two-way classical communication, even for stabilizer codes, and even in the antidegradable regime. For the depolarizing channel with $p$ slightly above $1/4$, the Choi state is isotropic with singlet fraction $1 - p > 1/2$, hence distillable by the recurrence-plus-hashing protocols of~\cite{bbpssw96,bdsw96,devetakwinter05}; via teleportation, the two-way assisted quantum capacity $Q_2$ is therefore strictly positive at rates bounded away from zero, while Corollary~\ref{cor:antideg} pins the unassisted stabilizer $\varepsilon$-capacity to zero for all $\varepsilon < 1$. In the proof, the exclusion of such assistance is visible in Lemma~\ref{lem:cost}: the flag---one sparse piece of receiver-side side information---is charged at its entropy, and it is precisely this accounting that unbounded interaction would circumvent.
\end{remark}

\subsection{Comparison with partial-transposition bounds}
\label{sec:pptcomparison}

\begin{remark}
\label{rem:ppt}
For $p \in [1/4, 1/2)$ the depolarizing channel $\cN_p$ is antidegradable, so Corollary~\ref{cor:antideg} gives a strong converse at rate $0$ for stabilizer codes. By contrast, all general strong converse rates for the unassisted quantum capacity known to us are based on Rains- or PPT-type relaxations~\cite{holevowerner01,tww17,wangfangduan19,bertawilde18}, and these are bounded away from $0$ in this regime, even after regularization. Indeed, the Choi state $\rho_p$ of $\cN_p$ is isotropic with singlet fraction $1-p > 1/2$ and hence has distillable entanglement $E_D(\rho_p) > 0$~\cite{bbpssw96,bdsw96}; since the Rains bound is an upper bound on distillable entanglement~\cite{rains99,rains01}, applying it to $\rho_p^{\otimes n}$ gives $\frac1n R(\rho_p^{\otimes n}) \geq E_D(\rho_p)$ for every $n$, and choosing maximally entangled inputs shows that the regularized Rains \emph{information} of $\cN_p$ is at least $E_D(\rho_p) > 0$. Within the stabilizer class, Theorem~\ref{thm:decay} is therefore strictly stronger than any bound derived from PPT relaxations in this window; outside the class, the two approaches are incomparable and complementary, since~\cite{tww17} applies to all codes.
\end{remark}

\subsection{Outlook}

Several directions suggest themselves. The exponent $E(\gamma) = \frac12 g^{-1}(\gamma/2)^2$ has not been optimized: the factor $2$ in $g^{-1}(\gamma/2)$ stems from crudely absorbing the continuity terms, and a more careful bookkeeping should give $E$ approaching $\frac12 g^{-1}(\gamma)^2$ for small $\gamma$; determining the optimal strong converse exponent for stabilizer codes, in analogy with the classical theory, is open. Beyond product noise, the blowing-up lemma holds for Markov measures with transportation-cost methods~\cite{marton96}, so Theorem~\ref{thm:master} should extend to Pauli noise with memory of Markov type; we have not pursued this. The extension to CWS codes appears to be the cheapest test of whether the boundary of the tractable class is where Section~\ref{sec:restriction} locates it. A different way to obtain the sharpness that Problem~\ref{prob:core} asks for is to let the \emph{channel} supply it: for erasure noise the pattern is handed to the receiver, branch fidelities are defined channel-side, and no-cloning forces the fidelities on complementary erasure patterns to be anti-correlated for \emph{every} code; combining this rigidity with the blowing-up lemma on the (classical, product) pattern measure might upgrade the almost-all-codes strong converse of~\cite{wildewinter14} to all codes, and we consider the erasure channel the most promising target beyond the Pauli class. Finally, Proposition~\ref{prop:nogo} shows that below capacity, constant fidelity can coexist with the absence of heavy cores; whether the same coexistence is possible at rates above the $n$-letter coherent information---where Theorem~\ref{thm:cores} forbids the cores but not, so far, the fidelity---is the strong converse question itself, and small-blocklength numerical searches for such coherently hedging codes above capacity would provide evidence for or against Problem~\ref{prob:core} and the conjecture alike.

\paragraph*{Acknowledgements:}
This manuscript was prepared in close collaboration with the large language model Claude Fable 5.0 and Opus 4.8 (Anthropic) and ChatGPT 5.6 (OpenAI). The mathematical content was developed in dialogue with Claude Fable: the author posed the problems, set the direction of the investigation, and supplied corrections and consistency checks, while the proof ideas and their execution throughout the paper originated with the model, which also drafted and typeset the manuscript and carried out the literature searches. The revisions were done with Opus. Bibliographic entries were checked against the original sources. ChatGPT acted as an adversarial referee and found several small inconsistencies in earlier drafts and made stylistic suggestions that improved the presentation. All definitions, statements, proofs and references have been verified by the author, who takes responsibility for any remaining errors. The author also thanks Andreas Winter for helpful comments on a previous draft. MT is supported by the NRF Investigatorship award (NRF-NRFI10-2024-0006). 

\appendix

\section{Extension to qudits}
\label{sec:qudits}

Let $d$ be a prime and $\cH_d = \mathbb{C}^d$ with computational basis $\{\ket j\}_{j \in \FF_d}$. Set $\omega := e^{2\pi i/d}$ and define the clock and shift operators $Z\ket j = \omega^j \ket j$, $X \ket j = \ket{j + 1}$. For $v = (a,b) \in \FF_d^2$ define the displacement operator $D_v := \tau^{ab} X^a Z^b$, where $\tau := \omega^{(d+1)/2}$ for odd $d$ (a $d$-th root of unity with $\tau^2 = \omega$) and $\tau := i$ for $d = 2$, recovering the convention of Section~\ref{sec:pauli}. On $n$ qudits, $D_v := D_{v_1} \otimes \cdots \otimes D_{v_n}$ for $v \in \FF_d^{2n}$. With the symplectic form $\langle u, v\rangle = \sum_i (a_i b_i' - b_i a_i') \bmod d$ of Section~\ref{sec:pauli} one has, using $ZX = \omega XZ$,
\begin{align}
  D_u D_v = \omega^{-\langle u,v \rangle} D_v D_u, \qquad
  D_u D_v = \eta(u,v)\, D_{u+v} \quad \text{with } |\eta(u,v)| = 1,
  \qquad D_v^\dagger = \eta'\, D_{-v},
\end{align}
and, for odd $d$, $\eta(u,v) = \tau^{-\langle u, v\rangle'}$ for an integer representative $\langle u,v\rangle'$ of the form, so that $D$ restricted to any isotropic subspace is an exact group homomorphism (as $\tau$ has order $d$); for $d = 2$ one uses the ordered products~\eqref{eq:Qdef} instead. A \emph{mixed displacement channel} (mixed-unitary Pauli channel) on one qudit is $\cN_i = \sum_{v\in\FF_d^2} \mu_i(v)\, D_v (\cdot) D_v^\dagger$, and we consider arbitrary products $\cN^{(n)}$ with product error measure $\mu$ on $\FF_d^{2n} = \prod_i \FF_d^2$. Stabilizer codes are defined exactly as in Section~\ref{sec:stabilizer}: an isotropic $\FF_d$-subspace $S \le \FF_d^{2n}$ of dimension $n - k$, sectors given by the $d^{n-k}$ joint eigenspaces (via the projectors $\Pi_s = \prod_j d^{-1}\sum_{m \in \FF_d} \omega^{-m s_j} Q_{g_j}^m$), the code space one full sector of dimension $D = d^k$, an arbitrary isometric encoder onto it, and an arbitrary decoder.

\begin{theorem}[Qudit master inequality and decay]
\label{thm:qudit}
Let $d$ be prime and let $\cN^{(n)}$ be a product of mixed displacement channels on $(\mathbb C^d)^{\otimes n}$. Define $g_d(\delta) := h(\delta) + \delta \log (d^2 - 1)$ for $\delta \in (0, 1 - d^{-2}]$, on which range $g_d$ increases onto $(0, 2\log d]$. Then for every $(n,k)$ stabilizer coding scheme (with $\log D = k \log d$) and every $\delta \in (0, 1 - d^{-2}]$,
\begin{align}
  \log D \;\leq\; \Ic(V) + n\, g_d(\delta) + 2 + 2 \log D\, \bigl(1 - \mu(\Gamma_{\delta n} A_\ML)\bigr) ,
\end{align}
with~\eqref{eq:masterblowup} unchanged. Consequently, if the code overshoots by $\log D - \Ic(V) \geq \gamma n$ with $\gamma \in (0, \log d]$, then $F \le C_d(\gamma)\exp(-n E_d(\gamma))$. Here $\delta_\gamma := g_d^{-1}(\gamma/2)$ and $E_d(\gamma) = \frac12 \delta_\gamma^2$ and $C_d(\gamma) := e^{\,n_0^{(d)}(\gamma) E_d(\gamma)}$ with
\begin{align}
  n_0^{(d)}(\gamma) := \Bigl\lceil \max\Bigl\{ \frac{8}{\gamma},\; \frac{2}{\delta_\gamma^2}\ln\frac{8\log d}{\gamma} \Bigr\} \Bigr\rceil .
\end{align}
Corollaries~\ref{cor:antideg}--\ref{cor:randomized} hold mutatis mutandis. Theorem~\ref{thm:cores} holds with $g_d$ in place of $g$ and $\gamma \in (0,\log d]$, but with the admissible core imperfection rescaled to $\theta \leq \gamma/(16 \log d)$ and the threshold to $n_1^{(d)}(\gamma) := \lceil \max\{16/\gamma,\, (2/\delta_\gamma^2)\ln(16\log d/\gamma)\}\rceil$.
\end{theorem}

\begin{proof}
The proof of the qubit case carries over with the following, purely notational, modifications. (i)~\emph{Sector structure:} in Lemma~\ref{lem:sectors}, involutions are replaced by unitaries of order $d$ and signs $(-1)^{(\cdot)}$ by characters $\omega^{(\cdot)}$; the destabilizers again permute the $d^{n-k}$ sectors transitively, so each has dimension $d^k$. In part (3), $\ell \ne 0$ admits $m'$ with $\langle m, m'\rangle = c \neq 0$, and $\tr_\cC U_\ell = \omega^{c} \tr_\cC U_\ell$ forces the trace to vanish; orthogonality reads $\tr_\cC(U_\ell^\dagger U_{\ell'}) = \eta\, D\,\delta_{\ell\ell'}$ with a phase $\eta$. In part (4) the character sum becomes $\sum_{\ell \in L} \omega^{\langle \ell, \ell'\rangle} = d^{2k}\, \delta_{\ell',0}$. (ii)~\emph{Fidelity formula:} Proposition~\ref{prop:fidelity} is verbatim, with $|L| = d^{2k} = D^2$ logical classes. (iii)~\emph{Blowing up:} Lemma~\ref{lem:blowup} is alphabet-independent and unchanged; the flag alphabet per site now has $d^2 - 1$ nonzero symbols, so the ball volume~\eqref{eq:volume} holds with $3$ replaced by $d^2 - 1$ and range $\delta \le 1 - d^{-2}$, giving $\log|G_\delta| \le n g_d(\delta) + 1$. (iv)~\emph{Flag correction:} on flag $u$ the decoder applies $D_{-u} = \eta\, D_u^\dagger$, shifting the error to $e - u \in A_\ML$. (v)~\emph{Accounting:} Lemmas~\ref{lem:cost} and~\ref{lem:oneshot} are dimension-agnostic, with $\log D = k \log d$ and $D^2 = d^{2k}$ logical classes in~\eqref{eq:oneshot}. The assembly in Theorems~\ref{thm:coremaster}, \ref{thm:master} and~\ref{thm:decay} is unchanged in form, using $Q^{(n)} \ge 0$ and $\log D \le n \log d$, so that $\gamma \in (0, \log d]$ and $\gamma/2 \le \log d \le 2 \log d = g_d(1 - d^{-2})$ lies in the range of $g_d$. One quantitative change does occur, and it is the only one. In the qubit proofs the prefactor $2k$ of the error term is bounded by $2n$; here the corresponding prefactor is $2\log D \le 2 n \log d$. Every allocation of the form ``$2k\,x \le \gamma n/8$'' therefore becomes ``$2 n \log d \cdot x \le \gamma n/8$'', which is what rescales the admissible core imperfection from $\gamma/16$ to $\gamma/(16\log d)$ and inserts the factor $\log d$ inside the logarithms in $n_0^{(d)}$ and $n_1^{(d)}$. The exponent $E_d(\gamma)$ is unaffected, since it is fixed by $g_d(\delta_\gamma) = \gamma/2$ alone.
\end{proof}

\begin{remark}[Prime powers]
\label{rem:primepower}
For $d = p^m$ a prime power, the same statement holds for stabilizer codes defined over the finite field $\FF_{d}$ in the standard way~\cite{gottesman97,ashikhminknill01,ketkar06}. The proof uses only the following four structural facts, all of which hold for the $\FF_d$-linear stabilizer formalism: a nice unitary error basis indexed by $\FF_d^{2n}$ with commutation governed by a nondegenerate symplectic (trace-)form; equal-dimensional joint eigenspaces of an isotropic subgroup, permuted transitively by displacements; the logical operators forming a complete orthogonal operator basis on a full sector, with the associated twirl; and submultiplicativity of the weight. It is worth saying why one should work over $\FF_d$ directly, rather than regarding each $\FF_{p^m}$-qudit as a block of $m$ qudits of prime dimension $p$ and appealing to Theorem~\ref{thm:qudit}. Under the latter view the noise acting on the $m$ subcoordinates of a single physical qudit is in general correlated, so the channel is \emph{not} a product over the $nm$ resulting sites and Theorem~\ref{thm:qudit} does not apply verbatim. Concentration is available at the level of the $n$ physical qudits, which is exactly what the $\FF_d$ formulation provides: the error measure is a product over those $n$ sites, with local alphabet $\FF_d^2$. Products of qudits of \emph{different} prime dimensions should likewise be covered---by the Chinese remainder theorem, subgroups of the corresponding product error group factor over the primes, and the blowing-up lemma is indifferent to the alphabet, which may vary from site to site---but we have not carried out the details. What we prove are the qubit case of Sections~\ref{sec:prelim}--\ref{sec:main} and the prime-dimensional case of Theorem~\ref{thm:qudit}; the prime-power statement rests on the structural facts listed above, for which we refer to~\cite{ashikhminknill01,ketkar06}.
\end{remark}

\section{Proof of the sector structure lemma}
\label{app:sectors}

\begin{proof}[Proof of Lemma~\ref{lem:sectors}]
(1) Each $Q_{g_j}$ is a Hermitian involution, so $(\Idop \pm Q_{g_j})/2$ are the orthogonal projectors onto its two eigenspaces, for the eigenvalues $\pm 1$. The $Q_{g_j}$ commute, hence so do all these projectors, and $\Pi_s$ is the orthogonal projector onto the subspace where $Q_{g_j}$ has eigenvalue $(-1)^{s_j}$ simultaneously for every $j$. Two distinct labels $s \ne s'$ disagree in some coordinate $j$, and the corresponding factors of $\Pi_s$ and $\Pi_{s'}$ are then projectors onto orthogonal eigenspaces of $Q_{g_j}$, so $\Pi_s \Pi_{s'} = 0$. Moreover
\begin{align}
  \Idop \;=\; \prod_{j=1}^{n-k}\Bigl( \frac{\Idop + Q_{g_j}}{2} + \frac{\Idop - Q_{g_j}}{2} \Bigr) \;=\; \sum_{s \in \FF_2^{n-k}} \Pi_s ,
\end{align}
by expanding the product and collecting the $2^{n-k}$ terms according to which sign is chosen in each factor. The dimension count uses~\eqref{eq:sectorshift} below: applying it with $e = t(s)$ gives $\sigma_{t(s)} \Pi_0 \sigma_{t(s)} = \Pi_s$, so all the $\Pi_s$ are unitarily conjugate and have equal rank. As there are $2^{n-k}$ of them and they sum to the identity on a $2^n$-dimensional space, each has rank $2^n/2^{n-k} = 2^k$.

(2) By~\eqref{eq:paulirules}, conjugating a Pauli operator by $\sigma_e$ leaves it unchanged if the two commute and flips its sign if they anticommute, so $\sigma_e \sigma_{g_j}\sigma_e = (-1)^{\langle e, g_j\rangle}\sigma_{g_j}$, using $\sigma_e^2 = \Idop$. Since $Q_{g_j} = \pm \sigma_{g_j}$ the same identity holds for $Q_{g_j}$, the sign cancelling on both sides. Conjugation by $\sigma_e$ therefore maps the $j$-th factor of $\Pi_s$ to
\begin{align}
  \frac{\Idop + (-1)^{s_j}\, \sigma_e Q_{g_j} \sigma_e}{2} \;=\; \frac{\Idop + (-1)^{s_j + \langle e, g_j\rangle}\, Q_{g_j}}{2} ,
\end{align}
and multiplying over $j$ gives
\begin{align}
 \sigma_e\, \Pi_s\, \sigma_e = \Pi_{s + \sigma^{\mathrm{syn}}(e)} .
 \label{eq:sectorshift}
\end{align}
So conjugation by $\sigma_e$ permutes the sectors by translation by the syndrome of $e$; it fixes every sector precisely when $\sigma^{\mathrm{syn}}(e) = 0$, that is, when $e \in S^\perp$.

(3) Let $\ell \in L$ and let $m \in S^\perp$ represent it. Then $\sigma^{\mathrm{syn}}(m) = 0$, so $\sigma_m$ preserves $\cC$ by (2); being unitary and mapping $\cC$ onto itself, it restricts to a unitary $U_\ell$ on $\cC$. To see that the restriction does not depend on the representative, note that two representatives of $\ell$ differ by some $g \in S$, and $\sigma_{m+g} = \pm\, \sigma_m Q_g$ by~\eqref{eq:paulirules} and~\eqref{eq:Qdef}. Now $\cC = \cH_0$ is by definition the subspace on which every $Q_{g_j}$ acts as $+1$, and $Q$ is a homomorphism, so $Q_g|_\cC = \Idop_\cC$ for every $g \in S$. Hence the two restrictions agree up to a sign, and the map $\rho \mapsto U_\ell \rho\, U_\ell^\dagger$ --- which is all we ever use --- is unambiguous.

Now suppose $\ell \neq 0$. By the nondegeneracy of the descended form on $L$ established above, there is $m' \in S^\perp$ with $\langle m, m'\rangle = 1$. Since $m' \in S^\perp$, conjugation by $\sigma_{m'}$ fixes $\Pi_0$ by~\eqref{eq:sectorshift}; since $\langle m,m'\rangle = 1$, it flips the sign of $\sigma_m$. Using cyclicity of the trace and $\sigma_{m'}^2 = \Idop$,
\begin{align}
 \tr\bigl(\Pi_0 \sigma_m \Pi_0\bigr) \;=\; \tr\bigl(\sigma_{m'}\, \Pi_0 \sigma_m \Pi_0\, \sigma_{m'}\bigr) \;=\; (-1)^{\langle m', m\rangle} \tr\bigl(\Pi_0 \sigma_m \Pi_0\bigr) \;=\; -\tr\bigl(\Pi_0 \sigma_m \Pi_0\bigr) ,
\end{align}
so this trace vanishes; it equals $\tr_\cC U_\ell$, whence $\tr_\cC U_\ell = 0$. For the orthogonality relation, $U_\ell^\dagger U_{\ell'} = \pm\, U_{\ell+\ell'}$ by~\eqref{eq:paulirules}: if $\ell = \ell'$ this is $\Idop_\cC$, of trace $D$, and otherwise $\ell + \ell' \neq 0$ and the trace vanishes by what we just proved.

(4) By (3) the $|L| = D^2$ operators $\{U_\ell\}_{\ell \in L}$ are pairwise orthogonal in the Hilbert--Schmidt inner product $(M,N) \mapsto \tr_\cC(M^\dagger N)$ on $\cB(\cC)$, hence linearly independent; since $\dim \cB(\cC) = D^2$ as well, they form an orthogonal basis. Write $M \in \cB(\cC)$ as $M = \sum_{\ell'} c_{\ell'} U_{\ell'}$. Conjugation permutes this basis up to signs, $U_\ell U_{\ell'} U_\ell^\dagger = (-1)^{\langle \ell, \ell'\rangle} U_{\ell'}$, where $\langle\cdot,\cdot\rangle$ now denotes the descended form on $L$ and the sign ambiguities of (3) cancel between $U_\ell$ and $U_\ell^\dagger$. Averaging over $\ell$,
\begin{align}
  \frac{1}{D^2}\sum_{\ell \in L} U_\ell M U_\ell^\dagger \;=\; \sum_{\ell'} c_{\ell'} \Bigl( \frac{1}{D^2}\sum_{\ell \in L} (-1)^{\langle \ell, \ell'\rangle} \Bigr) U_{\ell'} \;=\; c_0 U_0 ,
\end{align}
where we used the character sum $\sum_{\ell \in L}(-1)^{\langle \ell, \ell'\rangle} = D^2\, \delta_{\ell',0}$: for $\ell' \neq 0$ nondegeneracy makes $\ell \mapsto \langle \ell,\ell'\rangle$ a nonzero linear functional on $L$, so it vanishes on exactly half of $L$ and equals $1$ on the other half. Taking traces on both sides and using $\tr_\cC U_0 = D$ together with $\tr_\cC U_{\ell'} = 0$ for $\ell' \ne 0$ gives $c_0 = \tr_\cC(M)/D$, so the average equals $D^{-1}\tr_\cC(M)\, \Idop_\cC$. This is~\eqref{eq:twirl} for operators supported on $\cC$; a general $\rho \in \cB(\cH_R\otimes\cC)$ is a sum of products $A \otimes M$, on each of which the twirl acts only in the second factor, so~\eqref{eq:twirl} follows by linearity.
\end{proof}

% =====================================================================
% Bibliography
% All entries verified against original sources (arXiv/publisher pages),
% July 2026.
% =====================================================================
\section{Proof of the no-go proposition}
\label{app:nogo}

This appendix proves Proposition~\ref{prop:nogo}.

\begin{proof}
\emph{Construction.} Label the qubits $\{1,\dots,n\}$, with qubit $1$ a selector, $D_0 := \{2, \dots, k+1\}$ a data block and $P := \{k+2,\dots,n\}$ a padding block. Let $V_g : \mathbb C^{2^k} \to \cH_2^{\otimes(n-1)}$ be a stabilizer encoder on qubits $2,\dots,n$ achieving fidelity $F_g \to 1$ with some decoder $\cD_g$; such codes exist for $R < R_h$ by the hashing bound~\cite{bdsw96,devetakwinter05}. Define isometries $W_g\ket x := \ket 0_1 \otimes V_g \ket x$ and $W_j \ket x := \ket 1_1 \otimes \ket{x}_{D_0} \otimes \ket{0\cdots0}_P$, whose ranges are orthogonal, and the \emph{gambling encoder}
\begin{align}
  W := \sqrt c\, W_g + \sqrt{1-c}\, W_j .
\end{align}

\emph{Fidelity.} Let the decoder measure qubit $1$ in the computational basis; on outcome $0$ it applies $\cD_g$ to qubits $2,\dots,n$, on outcome $1$ it outputs the data block. The selector components of the two branches are orthogonal for every error, so the measurement produces no cross terms. If the error has no $X$-component on the selector, the good branch yields outcome $0$ with its full weight $c$ and is then decoded with its branch fidelity under $(V_g, \cD_g)$; all other contributions are nonnegative. Hence $F \geq c\, \mu_1(\{\Idop,Z\})\, F_g = c\,\bigl(1 - \tfrac{2p}{3}\bigr)F_g$, and $F_g \to 1$.

\emph{Branch splitting.} Let $(A, \widetilde\cD)$ be a $\theta$-core, with $\theta < 1-c$ as in the statement, and let $\{M_a\}$ be Kraus operators of $\widetilde\cD$. For any $e$, write $\psi_e := (\Idop\otimes\sigma_e W)\ket\Phi = \sqrt c\, g_e + \sqrt{1-c}\, j_e$ with the unit vectors $g_e := (\Idop\otimes\sigma_eW_g)\ket\Phi$ and $j_e := (\Idop\otimes\sigma_eW_j)\ket\Phi$. The triangle inequality in $\ell^2$ gives
\begin{align}
  \sqrt{f_e(W,\widetilde\cD)} \,=\, \Bigl\| \bigl(\langle\Phi|(\Idop\otimes M_a)\psi_e\rangle\bigr)_a \Bigr\|_2
  \;\leq\; \sqrt{c\, f^g_e} + \sqrt{(1-c)\, f^j_e}\, ,
\end{align}
where $f^{g}_e$ and $f^j_e$ denote the same expression with $g_e$ and $j_e$ in place of $\psi_e$. For $e \in A$ we have $f_e(W, \widetilde\cD) \ge 1-\theta$ and $f^g_e \le 1$, hence
\begin{align}
  f^j_e \;\geq\; \frac{\bigl(\sqrt{1-\theta} - \sqrt c\,\bigr)^2}{1-c} \;=\; v \;>\; 0 \,,
  \label{eq:junkgood}
\end{align}
where positivity holds because $\theta < 1-c$.

\emph{Rigidity.} The junk isometry $W_j$ is a stabilizer encoder: its code space is the full joint eigenspace of the stabilizer spanned by the $Z$-operators on $\{1\}\cup P$ (with eigenvalue $-1$ on qubit $1$ and $+1$ on $P$), the data block $D_0$ carries the logical algebra, and the syndrome reads off the $X$-components on $\{1\}\cup P$. Group the error patterns by this syndrome: for $x \in \FF_2^{n-k}$ let $F_x$ be the set of patterns whose $X$-components on $\{1\}\cup P$ equal $x$, and let $e_x \in F_x$ be the pattern with these $X$-components and nothing else. Any $e \in F_x$ differs from $e_x$ by an element of the symplectic complement of the junk stabilizer, which by Lemma~\ref{lem:sectors} preserves the junk code space and acts on it, up to sign, as the logical Pauli operator $\sigma_{w(e)}$ given by the $D_0$-component $w(e) \in \FF_2^{2k}$ of $e$; that is, $\sigma_{e-e_x} W_j = \pm W_j\, \sigma_{w(e)}$. Using~\eqref{eq:ricochet} to move $\sigma_{w(e)}$ to the reference system,
\begin{align}
  j_e \,=\, \pm\,(\Idop\otimes\sigma_{e_x} W_j \sigma_{w(e)})\ket\Phi \,=\, \pm\,(\sigma_{w(e)}^T \otimes \Idop)\, j_{e_x} \,,
\end{align}
so $f^j_e$ depends on $e$ only through $x$ and $w(e)$; write $f^j(x,w)$ for its common value. Distinct Pauli operators are orthogonal in the trace inner product, so Lemma~\ref{lem:rigidity}, applied with $\psi = j_{e_x}$ and the family $\{\sigma_w^T : w \in \FF_2^{2k}\}$, gives
\begin{align}
  \sum_{w \in \FF_2^{2k}} f^j(x, w) \;\leq\; 1 \qquad \text{for every } x .
\end{align}
Combined with~\eqref{eq:junkgood}, within each syndrome group at most $1/v$ of the $4^k$ logical classes $w$ can meet the core $A$.

\emph{Counting.} Fix $x$ and one admissible class $w$. The set $\{e \in F_x : w(e) = w\}$ fixes every component of $e$ on $D_0$ and the $X$-components on $\{1\}\cup P$, leaving the $Z$-components on $\{1\}\cup P$ unconstrained; since $\mu$ is a product measure, its mass is at most $\mu_{\max}^{\,k}\, \prod_{i\in\{1\}\cup P} \mu_1\bigl(\{v : v \text{ has } X\text{-component } x_i\}\bigr)$. Summing over the at most $1/v$ admissible classes and then over $x$, whose sum telescopes to $1$,
\begin{align}
 \mu(A) \;\le\; \frac1v\, \mu_{\max}^{\,k} \,.
\end{align}
\end{proof}


\begin{thebibliography}{99}

\bibitem{lloyd97} S. Lloyd, ``Capacity of the noisy quantum channel,'' \href{https://doi.org/10.1103/PhysRevA.55.1613}{Physical Review A \textbf{55}, 1613--1622 (1997)}.

\bibitem{shor02} P. W. Shor, ``The quantum channel capacity and coherent information,'' Lecture notes, MSRI Workshop on Quantum Computation (2002).

\bibitem{devetak05} I. Devetak, ``The private classical capacity and quantum capacity of a quantum channel,'' \href{https://doi.org/10.1109/TIT.2004.839515}{IEEE Transactions on Information Theory \textbf{51}, 44--55 (2005)}.

\bibitem{schumacher96} B. Schumacher, ``Sending entanglement through noisy quantum channels,'' \href{https://doi.org/10.1103/PhysRevA.54.2614}{Physical Review A \textbf{54}, 2614--2628 (1996)}.

\bibitem{bkn00} H. Barnum, E. Knill, and M. A. Nielsen, ``On quantum fidelities and channel capacities,'' \href{https://doi.org/10.1109/18.850671}{IEEE Transactions on Information Theory \textbf{46}, 1317--1329 (2000)}.

\bibitem{dss98} D. P. DiVincenzo, P. W. Shor, and J. A. Smolin, ``Quantum-channel capacity of very noisy channels,'' \href{https://doi.org/10.1103/PhysRevA.57.830}{Physical Review A \textbf{57}, 830--839 (1998)}.

\bibitem{smithsmolin07} G. Smith and J. A. Smolin, ``Degenerate quantum codes for Pauli channels,'' \href{https://doi.org/10.1103/PhysRevLett.98.030501}{Physical Review Letters \textbf{98}, 030501 (2007)}.

\bibitem{cubitt15} T. Cubitt, D. Elkouss, W. Matthews, M. Ozols, D. P\'erez-Garc\'ia, and S. Strelchuk, ``Unbounded number of channel uses may be required to detect quantum capacity,'' \href{https://doi.org/10.1038/ncomms7739}{Nature Communications \textbf{6}, 6739 (2015)}.

\bibitem{agk76} R. Ahlswede, P. G\'acs, and J. K\"orner, ``Bounds on conditional probabilities with applications in multi-user communication,'' \href{https://doi.org/10.1007/BF00535682}{Zeitschrift f\"ur Wahrscheinlichkeitstheorie und verwandte Gebiete \textbf{34}, 157--177 (1976)}.

\bibitem{marton86} K. Marton, ``A simple proof of the blowing-up lemma,'' \href{https://doi.org/10.1109/TIT.1986.1057176}{IEEE Transactions on Information Theory \textbf{32}, 445--446 (1986)}.

\bibitem{marton96} K. Marton, ``Bounding $\bar d$-distance by informational divergence: a method to prove measure concentration,'' \href{https://doi.org/10.1214/aop/1039639365}{Annals of Probability \textbf{24}, 857--866 (1996)}.

\bibitem{mcdiarmid89} C. McDiarmid, ``On the method of bounded differences,'' in \emph{Surveys in Combinatorics}, \href{https://doi.org/10.1017/CBO9781107359949.008}{London Mathematical Society Lecture Note Series \textbf{141}, 148--188}, Cambridge University Press (1989).

\bibitem{ck11} I. Csisz\'ar and J. K\"orner, \emph{Information Theory: Coding Theorems for Discrete Memoryless Systems}, 2nd ed., \href{https://doi.org/10.1017/CBO9780511921889}{Cambridge University Press (2011)}.

\bibitem{morganwinter14} C. Morgan and A. Winter, ``\,`Pretty strong' converse for the quantum capacity of degradable channels,'' \href{https://doi.org/10.1109/TIT.2013.2288971}{IEEE Transactions on Information Theory \textbf{60}, 317--333 (2014)}.

\bibitem{tww17} M. Tomamichel, M. M. Wilde, and A. Winter, ``Strong converse rates for quantum communication,'' \href{https://doi.org/10.1109/TIT.2016.2615847}{IEEE Transactions on Information Theory \textbf{63}, 715--727 (2017)}

\bibitem{wangfangduan19} X. Wang, K. Fang, and R. Duan, ``Semidefinite programming converse bounds for quantum communication,'' \href{https://doi.org/10.1109/TIT.2018.2874031}{IEEE Transactions on Information Theory \textbf{65}, 2583--2592 (2019)}

\bibitem{bertawilde18} M. Berta and M. M. Wilde, ``Amortization does not enhance the max-Rains information of a quantum channel,'' \href{https://doi.org/10.1088/1367-2630/aac153}{New Journal of Physics \textbf{20}, 053044 (2018)}

\bibitem{wildewinter14} M. M. Wilde and A. Winter, ``Strong converse for the quantum capacity of the erasure channel for almost all codes,'' Proceedings of the 9th Conference on the Theory of Quantum Computation, Communication and Cryptography (TQC 2014), \href{https://doi.org/10.4230/LIPIcs.TQC.2014.52}{LIPIcs vol.~27, 52--66 (2014)}

\bibitem{kdww21} E. Kaur, S. Das, M. M. Wilde, and A. Winter, ``Resource theory of unextendibility and nonasymptotic quantum capacity,'' \href{https://doi.org/10.1103/PhysRevA.104.022401}{Physical Review A \textbf{104}, 022401 (2021)}.

\bibitem{khanianhirche25} Z. B. Khanian and C. Hirche, ``On strong converse bounds for the private and quantum capacities of anti-degradable channels,'' \href{https://arxiv.org/abs/2507.15661}{arXiv:2507.15661} (2025).

\bibitem{holevowerner01} A. S. Holevo and R. F. Werner, ``Evaluating capacities of bosonic Gaussian channels,'' \href{https://doi.org/10.1103/PhysRevA.63.032312}{Physical Review A \textbf{63}, 032312 (2001)}.

\bibitem{rains99} E. M. Rains, ``Bound on distillable entanglement,'' \href{https://doi.org/10.1103/PhysRevA.60.179}{Physical Review A \textbf{60}, 179--184 (1999)}.

\bibitem{rains01} E. M. Rains, ``A semidefinite program for distillable entanglement,'' \href{https://doi.org/10.1109/18.959270}{IEEE Transactions on Information Theory \textbf{47}, 2921--2933 (2001)}.

\bibitem{bbpssw96} C. H. Bennett, G. Brassard, S. Popescu, B. Schumacher, J. A. Smolin, and W. K. Wootters, ``Purification of noisy entanglement and faithful teleportation via noisy channels,'' \href{https://doi.org/10.1103/PhysRevLett.76.722}{Physical Review Letters \textbf{76}, 722--725 (1996)}.

\bibitem{bdsw96} C. H. Bennett, D. P. DiVincenzo, J. A. Smolin, and W. K. Wootters, ``Mixed-state entanglement and quantum error correction,'' \href{https://doi.org/10.1103/PhysRevA.54.3824}{Physical Review A \textbf{54}, 3824--3851 (1996)}.

\bibitem{devetakwinter05} I. Devetak and A. Winter, ``Distillation of secret key and entanglement from quantum states,'' \href{https://doi.org/10.1098/rspa.2004.1372}{Proceedings of the Royal Society A \textbf{461}, 207--235 (2005)}.

\bibitem{devetakshor05} I. Devetak and P. W. Shor, ``The capacity of a quantum channel for simultaneous transmission of classical and quantum information,'' \href{https://doi.org/10.1007/s00220-005-1317-6}{Communications in Mathematical Physics \textbf{256}, 287--303 (2005)}.

\bibitem{fannes73} M. Fannes, ``A continuity property of the entropy density for spin lattice systems,'' \href{https://doi.org/10.1007/BF01646490}{Communications in Mathematical Physics \textbf{31}, 291--294 (1973)}.

\bibitem{audenaert07} K. M. R. Audenaert, ``A sharp continuity estimate for the von Neumann entropy,'' \href{https://doi.org/10.1088/1751-8113/40/28/S18}{Journal of Physics A: Mathematical and Theoretical \textbf{40}, 8127--8136 (2007)}.

\bibitem{alickifannes04} R. Alicki and M. Fannes, ``Continuity of quantum conditional information,'' \href{https://doi.org/10.1088/0305-4470/37/5/L01}{Journal of Physics A: Mathematical and General \textbf{37}, L55--L57 (2004)}.

\bibitem{winter16} A. Winter, ``Tight uniform continuity bounds for quantum entropies: conditional entropy, relative entropy distance and energy constraints,'' \href{https://doi.org/10.1007/s00220-016-2609-8}{Communications in Mathematical Physics \textbf{347}, 291--313 (2016)}.

\bibitem{alhejjismith20} M. A. Alhejji and G. Smith, ``A tight uniform continuity bound for equivocation,'' in Proceedings of the IEEE International Symposium on Information Theory (ISIT), 2270--2274 (2020), \href{https://arxiv.org/abs/1909.00787}{arXiv:1909.00787}.

\bibitem{wilde20} M. M. Wilde, ``Optimal uniform continuity bound for conditional entropy of classical--quantum states,'' \href{https://doi.org/10.1007/s11128-019-2563-4}{Quantum Information Processing \textbf{19}, 61 (2020)}.

\bibitem{bertalamitomamichel25} M. Berta, L. Lami, and M. Tomamichel, ``Continuity of entropies via integral representations,'' \href{https://doi.org/10.1109/TIT.2025.3527858}{IEEE Transactions on Information Theory \textbf{71}, 1896--1908 (2025)}

\bibitem{audenaertetal25} K. Audenaert, B. Bergh, N. Datta, M. G. Jabbour, \'A. Capel, and P. Gondolf, ``Continuity bounds for quantum entropies arising from a fundamental entropic inequality,'' \href{https://doi.org/10.1109/TIT.2025.3586478}{IEEE Transactions on Information Theory \textbf{71}, 7029--7038 (2025)}

\bibitem{fvdg99} C. A. Fuchs and J. van de Graaf, ``Cryptographic distinguishability measures for quantum-mechanical states,'' \href{https://doi.org/10.1109/18.761271}{IEEE Transactions on Information Theory \textbf{45}, 1216--1227 (1999)}.

\bibitem{gottesman97} D. Gottesman, \emph{Stabilizer Codes and Quantum Error Correction}, PhD thesis, Caltech (1997), \href{https://arxiv.org/abs/quant-ph/9705052}{arXiv:quant-ph/9705052}.

\bibitem{ashikhminknill01} A. Ashikhmin and E. Knill, ``Nonbinary quantum stabilizer codes,'' \href{https://doi.org/10.1109/18.959288}{IEEE Transactions on Information Theory \textbf{47}, 3065--3072 (2001)}.

\bibitem{niwalee25} R. Niwa and J. Y. Lee, ``Coherent information for CSS codes under decoherence,'' \href{https://doi.org/10.1103/PhysRevA.111.032402}{Physical Review A \textbf{111}, 032402 (2025)}

\bibitem{kann26} T. Kann, M. R. Bloch, S. Kudekar, and R. Urbanke, ``Stabilizer-code channel transforms beyond repetition codes for improved hashing bounds,'' \href{https://arxiv.org/abs/2601.15505}{arXiv:2601.15505} (2026).

\bibitem{gabidulin67} E. M. Gabidulin, ``Limits for the decoding error probability when linear codes are used in memoryless channels,'' Problems of Information Transmission \textbf{3}, 43--48 (1967).

\bibitem{agarwaletal26} A. Agarwal, A. R. Kalra, S. Lee, D. Leung, L. Schaeffer, P. Sinha, and G. Smith, ``Enhanced quantum capacity thresholds from symmetry,'' \href{https://arxiv.org/abs/2605.09138}{arXiv:2605.09138} (2026).

\bibitem{hamada05} M. Hamada, ``Information rates achievable with algebraic codes on quantum discrete memoryless channels,'' IEEE Transactions on Information Theory \textbf{51}, 4263--4277 (2005), \href{https://arxiv.org/abs/quant-ph/0207113}{arXiv:quant-ph/0207113}.

\bibitem{ketkar06} A. Ketkar, A. Klappenecker, S. Kumar, and P. K. Sarvepalli, ``Nonbinary stabilizer codes over finite fields,'' \href{https://doi.org/10.1109/TIT.2006.883612}{IEEE Transactions on Information Theory \textbf{52}, 4892--4914 (2006)}

\bibitem{bruss98} D. Bru{\ss}, D. P. DiVincenzo, A. Ekert, C. A. Fuchs, C. Macchiavello, and J. A. Smolin, ``Optimal universal and state-dependent quantum cloning,'' \href{https://doi.org/10.1103/PhysRevA.57.2368}{Physical Review A \textbf{57}, 2368--2378 (1998)}.

\bibitem{ssrw17} D. Sutter, V. B. Scholz, A. Winter, and R. Renner, ``Approximate degradable quantum channels,'' \href{https://doi.org/10.1109/TIT.2017.2754268}{IEEE Transactions on Information Theory \textbf{63}, 7832--7844 (2017)}

\bibitem{lls18} F. Leditzky, D. Leung, and G. Smith, ``Quantum and private capacities of low-noise channels,'' \href{https://doi.org/10.1103/PhysRevLett.120.160503}{Physical Review Letters \textbf{120}, 160503 (2018)}

\bibitem{dklp02} E. Dennis, A. Kitaev, A. Landahl, and J. Preskill, ``Topological quantum memory,'' \href{https://doi.org/10.1063/1.1499754}{Journal of Mathematical Physics \textbf{43}, 4452--4505 (2002)}.

\bibitem{bcgst02} H. Barnum, C. Cr\'epeau, D. Gottesman, A. Smith, and A. Tapp, ``Authentication of quantum messages,'' \href{https://doi.org/10.1109/SFCS.2002.1181969}{Proceedings of the 43rd Annual IEEE Symposium on Foundations of Computer Science (FOCS 2002), pp. 449--458 (2002)}.

\bibitem{osbornewinter09} T. J. Osborne and A. Winter, ``A quantum generalisation of Talagrand's inequality,'' online research note (2009), available at \url{https://tjoresearchnotes.wordpress.com/2009/02/13/a-quantum-generalisation-of-talagrands-inequality/}.

\bibitem{rouzewirthzhang24} C. Rouz\'e, M. Wirth, and H. Zhang, ``Quantum Talagrand, KKL and Friedgut's theorems and the learnability of quantum Boolean functions,'' \href{https://doi.org/10.1007/s00220-024-04981-0}{Communications in Mathematical Physics \textbf{405}, 95 (2024)}

\bibitem{changli26} F. Chang and P. Li, ``Quantum Talagrand-type inequalities via variance decay,'' \href{https://arxiv.org/abs/2601.01900}{arXiv:2601.01900} (2026).

\bibitem{gowershatami17} W. T. Gowers and O. Hatami, ``Inverse and stability theorems for approximate representations of finite groups,'' Sbornik: Mathematics \textbf{208}, 1784--1817 (2017), \href{https://arxiv.org/abs/1510.04085}{arXiv:1510.04085}

\bibitem{benyoreshkov10} C. B\'eny and O. Oreshkov, ``General conditions for approximate quantum error correction and near-optimal recovery channels,'' \href{https://doi.org/10.1103/PhysRevLett.104.120501}{Physical Review Letters \textbf{104}, 120501 (2010)}.

\bibitem{cws09} A. Cross, G. Smith, J. A. Smolin, and B. Zeng, ``Codeword stabilized quantum codes,'' \href{https://doi.org/10.1109/TIT.2008.2008136}{IEEE Transactions on Information Theory \textbf{55}, 433--438 (2009)}

\bibitem{rhss97} E. M. Rains, R. H. Hardin, P. W. Shor, and N. J. A. Sloane, ``A nonadditive quantum code,'' \href{https://doi.org/10.1103/PhysRevLett.79.953}{Physical Review Letters \textbf{79}, 953--954 (1997)}.

\bibitem{ycl08} S. Yu, Q. Chen, C. H. Lai, and C. H. Oh, ``Nonadditive quantum error-correcting code,'' \href{https://doi.org/10.1103/PhysRevLett.101.090501}{Physical Review Letters \textbf{101}, 090501 (2008)}

\bibitem{tbr16} M. Tomamichel, M. Berta, and J. M. Renes, ``Quantum coding with finite resources,'' \href{https://doi.org/10.1038/ncomms11419}{Nature Communications \textbf{7}, 11419 (2016)}.

\end{thebibliography}
\end{document}